\documentclass[11pt]{article}
\usepackage[margin=1in]{geometry}
\usepackage{amsthm}
\usepackage{amsmath}
\usepackage{amsfonts}
\usepackage{amssymb}
\usepackage{graphicx}
\usepackage{subfig}
\usepackage{caption}
\usepackage{algorithm}
\usepackage{url}
\usepackage{mathptmx}
\usepackage{epsfig}
\usepackage{wrapfig}
\usepackage{color}
\usepackage{multirow}
\usepackage{epstopdf}
\usepackage{algorithm}
\usepackage{algpseudocode}
\usepackage{comment}
\usepackage{authblk}
\usepackage{paralist}

\usepackage[all]{xy}

\usepackage{wrapfig}

\newcommand{\R}{\mathbb{R}}

\newcommand{\eps}{\epsilon}
\renewcommand{\input}{P}

\newcommand{\nerve}{\mathrm{Nrv}}

\newcommand{\Clos}{\mathrm{Clos}}

\newcommand{\setP}{\mathcal{P}}
\newcommand{\setQ}{\mathcal{Q}}

\newcommand{\im}{\mathrm{im}}

\newcommand{\thick}{\mathcal{T}}
\newcommand{\voronoi}{\mathrm{Vor}}

\newcommand{\critical}[1]{\mathrm{crit}(#1)}

\newcommand{\distance}{\mathrm{dist}}

\newcommand{\ignore}[1]{}

\newtheorem{theorem}{Theorem}
\newtheorem{lemma}[theorem]{Lemma}

\newtheorem{definition}[theorem]{Definition}

\title{The Offset Filtration of Convex Objects\thanks{
Work by D.H.\ and D.S.\ has been supported in part by the Israel Science Foundation 
(grant no. 1102/11), by the German-Israeli Foundation 
(grant no. 1150-82.6/2011), and by the Hermann Minkowski--Minerva
Center for Geometry at Tel Aviv University. 
M.K. acknowledges support by the Max Planck Center of Visual 
Computing and Communication.
}
}

\author[1]{Dan Halperin\thanks{danha@post.tau.ac.il}}
\author[2]{Michael Kerber\thanks{mkerber@mpi-inf.mpg.de}}
\author[1]{Doron Shaharabani\thanks{doron.s@hotmail.com}}
\affil[1]{Tel Aviv University, Tel Aviv, Israel}
\affil[2]{Max Planck Institute for Informatics, Saarbr\"ucken, Germany}

\date{}

\begin{document}

\maketitle
\begin{abstract}
We consider offsets of a union of convex objects.
We aim for a filtration, a sequence of nested
cell complexes, that captures the topological evolution of 
the offsets for increasing radii. 
We describe methods to compute a filtration based on the
Voronoi partition with respect to the given convex objects.
We prove that, in two and three dimensions, 
the size of the filtration 
is proportional to the size of the Voronoi diagram.
Our algorithm runs in $\Theta(n \log{n})$ in the $2$-dimensional case 
and in expected time $O(n^{3 + \eps})$, for any $\eps > 0$, in the 
$3$-dimensional case.
Our approach is inspired by alpha-complexes for point sets, 
but requires more involved machinery and analysis primarily since
Voronoi regions of general convex objects do not form a
good cover.
We show by experiments that our approach results in a similarly fast
and topologically more stable method for computing a filtration compared
to approximating the input by point samples.
\end{abstract}

\section{Introduction}
\paragraph{Motivation}
The theory of \emph{persistent homology} has led to a new way
of understanding data through its topological properties, commonly
referred as \emph{topological data analysis}.
The most common setup assumes that the data is given as a finite set of points
and analyzes the sublevel sets of the distance function to the point set.
An equivalent formulation is to take offsets of the point sets with increasing
offset parameter and to study the changes in the hole structure of the shape
obtained by the union of the offset balls; see Figure~\ref{fig:barcode_example} 
for an illustration and informal description. Notice that we
postpone the exposition of formal topology background to
the next section.

We pose the question how to generalize the default framework for point sets to
more general input shapes. While there is no theoretical obstacle to consider
distance functions from shapes rather than points (at least for reasonably ``nice''
shapes), it raises computational questions: How can critical points of that
distance functions be computed efficiently? And how can the topological
information be encoded in a combinatorial structure of small size?

With the wealth of applications of persistence of point set data,
and together with the challenges raised by the extension
from point sets to sets of convex objects, we believe that
the latter is a logical next step of
investigation. Our attention to this problem originates
from the increasingly popular application of 3D printing. A
common problem in this context is that often available models of
shapes contain features that complicate the printing
process, or turn it impossible altogether.  A
ubiquitous example is the presence of thin features which
may easily break, and call for thickening. One work-around
is to offset the model by a small value to stabilize it,
but the optimal offset parameter is unclear, as it should
get rid of many spurious features of the model without
introducing too many new ones. Moreover, one would prefer
{\it local thickening}~\cite{DBLP:journals/tog/StavaVBCM12}, and
possibly thickening by different offset size in different
parts of the model. A by-product of our work here is a step toward
automatically detecting target regions for local thickening
that do not incur spurious artifacts. 
Persistent homology provides a \emph{barcode} which constitutes
a summary of the hole structure of the offset shape for any parameter value
(Figure~\ref{fig:barcode_example})
which is clearly helpful for the choice of a good offset value.
We are especially interested in an exact method in this context because
any approximate barcode 
(obtained, for instance, by approximating the shapes by point set)
introduces artificial topological noise
which are difficult to discern from short-ranged features. 
This, in turn, makes it even more difficult to choose a suitable offset radius.

\paragraph{Problem definition and contribution}
We design, analyze, implement, and experimentally evaluate algorithms for computing persistence barcodes
of convex input objects. 
More precisely, we concentrate on the problem of computing a \emph{filtration}, a sequence
of nested combinatorial cell complexes that undergoes the same topological changes as the offset shapes.
Since the input objects are convex, the nerve theorem asserts that the intersection patterns of the offsets
(called the \emph{nerve}) reveal the entire topological information.
This leads to the generalization of \emph{\v{C}ech} filtrations from point sets to our scenario.
The resulting filtration has a size of $O(n^{d+1})$, where $n$ is the number of input objects contained in $d$-dimensional
Euclidean space.
This size is already problematic for small $d$ and 
a natural idea to reduce its size is to consider \emph{restricted offsets}, that is, intersecting the offset of an input object
to the \emph{Voronoi region}, the portion of the space which is closest to the object. 
This approach is again inspired by the analogue
case of point sets, where \emph{alpha-complexes} are preferred over the \v{C}ech complexes for small dimensions.
However, the approach for point sets does \emph{not} directly carry over to arbitrary convex objects:
Voronoi regions of convex objects are not necessarily convex and can intersect in non-contractible patterns
which prevents the application of the nerve theorem.

\begin{figure}
\centering
\includegraphics[width=2.2cm]{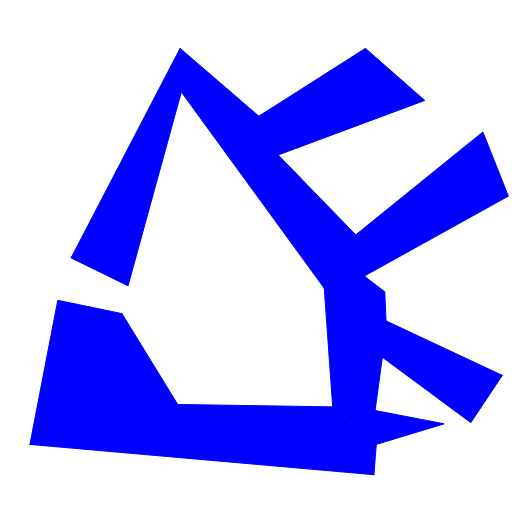}
\includegraphics[width=2.2cm]{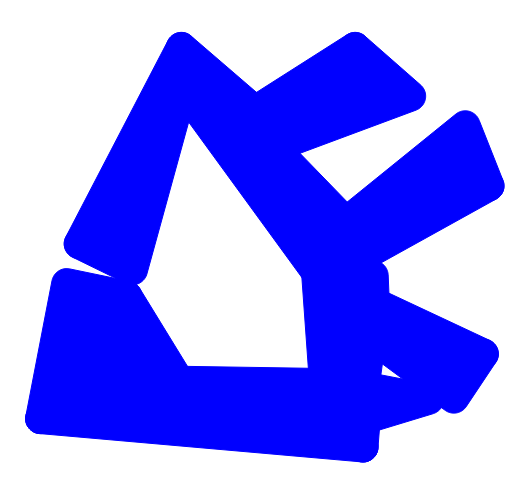}
\includegraphics[width=2.2cm]{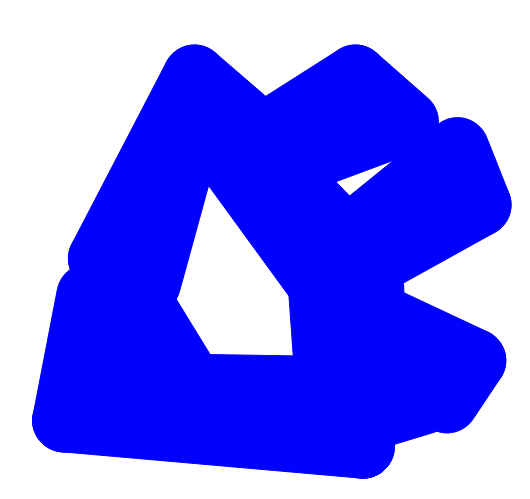}
\includegraphics[width=2.2cm]{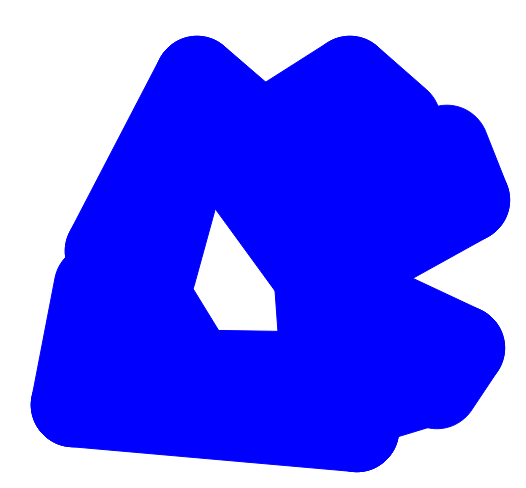}
\includegraphics[width=4.4cm]{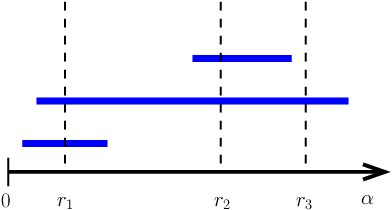}
\caption{From left to right, we see an example shape, three offsets
with increasing radii $r_1<r_2<r_3$, and the $1$-barcode of the shape. 
While being simply-connected initially, two holes have been formed at radius $r_1$,
one of which disappears for a slightly larger offset value while the other one \emph{persists}
for a large range of scales. At $r_2$, we see the formation of another rather short-lived hole.
The barcode summarizes these facts by displaying one bar per hole.
The bar spans over the range of offset radii for which the hole is present.
We can read off the number of holes for radius $\alpha$ by counting the number of bars that have $\alpha$
in their $x$-range.
}
\label{fig:barcode_example}
\vspace{-3mm}
\end{figure}

Our first result is that in $\R^2$, the non-contractibility does not really cause problems:
the barcode of convex polygons is encoded in the barcode of the nerve of their restricted offsets,
despite the presence of non-contractible intersections.
The analogue statement in $\R^3$ is not true.
The result implies the existence of a linear-sized filtration, as opposed to the cubical
size obtained from using the unrestricted nerves. The filtration can be computed in time 
$O(n \log{n})$, ruled by the computation time of the Voronoi diagram.
While the proof ultimately still relies on the nerve theorem, it requires a deeper investigation
of the structure of Voronoi diagrams of convex objects.
Moreover, it requires a slight generalization of a result in~\cite{co-towards}
(see Theorem~\ref{thm:nerve_inclusion_commute} below),
showing that the nerve isomorphism commutes with inclusions, to the case of filtrations connected
by certain simplicial maps.

Our second result is a general construction of a cell complex with the desired barcode in three dimensions. 
Our construction scheme computes the Voronoi diagram of the input sites as a preprocessing step
and cuts (subdivides) the lower-dimensional cells of the Voronoi diagrams into smaller pieces in a controlled way.
The resulting refinement of the Voronoi diagram gives rise to a dual cell complex whose size is asymptotically
equal to the complexity of the Voronoi diagram of the input sites. As the latter is known to be bounded by $O(n^{3+\eps})$,
our filtration is significantly smaller than $O(n^4)$, as obtained by a \v{C}ech-like filtration. 
The time for computing the filtration is bounded from above by $O(n^{3+\eps})$.
The correctness proof works by (conceptually) ``thicken up'' lower-dimensional cells
of the Voronoi diagram to obtain a good cover of the space, for which the nerve theorem applies.

We have implemented our algorithm for polygons using the \textsc{Cgal} library and report on extensive experimental
evaluation. In particular, we compare our approach with the natural alternative to replace
the input polygons with sufficiently dense point samples. 
Although the point sample approach yields very close approximations to the exact barcode in a comparable running time,
we demonstrate that the approximation error induced by the sampling results in additional noise on a large range of scales
and therefore makes the topological analysis of the offset filtration more difficult.

\paragraph{Related work}
Since its introduction in~\cite{elz-topological}, persistent homology has become
an active area of research, including theoretical, algorithmic, and application results;
we refer to the textbook~\cite{eh-computational} and the surveys~\cite{em-survey, carlsson-topology} for an overview.
The information gathered by persistence is usually displayed either in terms of a barcode (as in this work)
or, equivalently, via a \emph{persistence diagram}~\cite{eh-computational}.

The textbook~\cite{eh-computational} describes the most common approaches for computing filtrations of point sets, 
including \v{C}ech- and alpha-complexes mentioned above. 
Another common construction
is the \emph{Vietoris-Rips} complex which approximates the \v{C}ech complex in the sense that it is nested between two \v{C}ech complexes
on similar scales; precisely, the Rips complex at scale $\alpha$
contains the \v{C}ech complex at scale $\alpha$, and is contained in the
\v{C}ech complex at scale $\sqrt{2}\alpha$.
However, it is easy to see that this property does not carry over to the case of arbitrary convex objects.

Topological methods for shape analysis have been extensively studied: a commonly used concept are \emph{Reeb graphs} which yield a skeleton representing the connectivity 
of the shape and can be seen as a special case of persistent homology in dimension $0$; see~\cite{reeb-survey} for ample applications.
The full theory of persistent homology has also been applied to various tasks in shape analysis, including shape segmentation~\cite{socg-persistence}
and partial shape similarity~\cite{fl-persistent}. While these works study the \emph{intrinsic} properties of a shape through descriptor functions
independent of the embedding, our problem setup rather asks about \emph{extrinsic} properties, that is, how the shape is embedded in ambient space.

Voronoi diagrams are one of the most basic objects in computational geometry~\cite{ak-chapter,dutch,fortune-voronoi}. Efficient algorithms for the case
of point sets have been designed and implemented in 2D and 3D~\cite{pt-3d,shewchuk-delaunay,yvinec-2d} and higher dimensions~\cite{hb-efficient}.
In the plane, generalizations to convex objects~\cite{ky-voronoi} and line segments and circular arcs~\cite{held-vroni,yap-vd} have been presented.
Generalizations in three dimensions include Voronoi diagrams of (infinite) lines~\cite{hsh-constructing}
and approximating the Voronoi diagram of one polyhedron~\cite{milenkovic-robust}. However, an exact and efficient method
for a set of (convex) polyhedra in $\R^3$ is still missing; we refer to~\cite{ysl-towards} for a discussion of the difficulties.

\paragraph{Outline} 
We describe barcodes of convex objects, generalizing \v{C}ech complexes, in Section~\ref{sec:topology},
introducing basic topological concepts. We define our generalization of alpha-complexes in Section~\ref{sec:restricted_barcodes} and 
show how it leads to more efficient filtrations for $d=2$ (Section~\ref{sec:restricted_barcodes_2d}) and $d=3$ (Section~\ref{sec:higher_dim}).
We report on experimental evaluations for the planar case in Section~\ref{sec:experiments}.
We conclude in Section~\ref{sec:conclusion}.

\section{Topological background}
\label{sec:topology}
\label{sec:barcodes_of_shapes}
We review standard notation in persistent homology and dualizations of set covers through nerves.
We assume familiarity with basic topological notions, in particular simplicial complexes and homology groups;
the necessary background is covered by the textbook~\cite{eh-computational} and in more detail by~\cite{hatcher,munkres}.

\paragraph{Persistent homology}
A \emph{persistence module} is sequence of vector spaces $(V_\alpha)_{\alpha\geq 0}$ with linear maps $F_{\alpha,\alpha'}:V_\alpha\rightarrow V_{\alpha'}$ for $\alpha\leq\alpha'$
that satisfy $F_{\alpha',\alpha''}\circ F_{\alpha,\alpha'} = F_{\alpha,\alpha''}$ and $F_{\alpha,\alpha}$ is the identity function on $V_\alpha$~\cite{ccggo-proximity}.
We say that $\alpha$ is a \emph{critical value} if $V_{\alpha-\eps}$ is not isomorphic to $V_{\alpha}$ for $\eps>0$.
We assume the usual \emph{tameness conditions} that each $V_\alpha$ has finite rank, 
and the number of critical values is finite.
A generator (basis element) $\gamma$ of $V_\alpha$ is \emph{born at $\alpha$} if $\gamma \notin\im F_{\alpha-\eps,\alpha}$ for any $\eps>0$.
A generator $\gamma$ born at $\alpha$ \emph{dies at $\beta$}, if $\gamma\in\im F_{\alpha-\eps,\beta}$, but $\gamma\notin \im F_{\alpha-\eps,\beta-\eps}$.
In this way, every generator in $V_\alpha$ is assigned to a birth-death interval 
and the length of the birth-death interval is called the \emph{persistence} of the generator.
The \emph{barcode} of the persistence module is the set of these intervals.

We call a collection of spaces $(Q_\alpha)_{\alpha\geq 0}$ with the property that $Q_{\alpha}\subseteq Q_{\alpha'}$ whenever $\alpha\leq\alpha'$
a \emph{filtration induced by inclusion}.
A standard way to obtain a persistence module is to apply the \emph{homology functor} on such a filtration, that is, 
to consider $(H_p(Q_\alpha))_{\alpha\geq 0}$ where $H_p(\cdot)$ is the $p$-th homology group over an arbitrary fixed base field, with $p\geq 0$.
Indeed, the inclusion map from $Q_\alpha$ to $Q_{\alpha'}$ induces a map $F_{\alpha,\alpha'}H_p(Q_\alpha)\rightarrow H_p(Q_{\alpha'})$ satisfying
the required properties. This construction works generally for filtrations of simplicial complexes (using simplicial homology),
of CW-complexes (using cellular homology) and of subsets of $\R^d$ (using singular homology).

For simplicial complexes, we will also use a more general construction: a \emph{simplicial map} $F:S_1\rightarrow S_2$ between
two simplicial complexes $S_1$ and $S_2$ is a map induced by mapping vertices of $S_1$ to vertices of $S_2$. We call a sequence
of simplicial complexes $(S_\alpha)_{\alpha\geq 0}$ together with simplicial maps $F_{\alpha,\alpha'}$ 
a \emph{filtration induced by simplicial maps}. Such a filtration gives rise to a persistence module, and thus a barcode, in the same
way as in the inclusion case, by applying the functor $H_p(\cdot)$.
We call the barcode of the persistence module of a filtration using $H_p(\cdot)$ the \emph{$p$-barcode} of the filtration

\paragraph{Nerves}
Let $\setP:=\{P^1,\ldots,P^n\}$ be a collection of non-empty sets in a common domain.
The \emph{underlying space} is defined as $|\setP|:=\bigcup_{i=1,\ldots,n} P^i$.
We call a non-empty subset $\{P^{i_1},\ldots,P^{i_k}\}\subseteq\setP$ \emph{intersecting}, 
if $\bigcap_{j=1}^{k}P^{i_j}\neq\emptyset$.
The \emph{nerve} $\nerve(\setP)$ of $\setP$ is the collection of all intersecting subsets.
It is clear by definition that every singleton set $\{P^i\}$ is in the nerve,
and that any non-empty subset of an intersecting set is intersecting.
The latter property implies that the nerve is a \emph{simplicial complex}:
the singleton sets $\{P_i\}$ are the \emph{vertices} of that complex.
We call $\setP$ a \emph{good cover} if all sets in the collection
are closed and triangulable, and any intersecting subset yields a contractible intersection.
For example, any collection of closed convex sets forms a good cover.
\begin{theorem}[Nerve Theorem]
\label{thm:nerve_theorem}
If $\setP$ is a good cover, $|\setP|$ is homotopically equivalent to $\nerve(\setP)$.
In particular, $H_p(|\setP|)=H_p(\nerve(\setP))$ for all $p\geq 0$.
\end{theorem}

The following lemma is a slightly modified version of~\cite[Lemma 3.4]{co-towards}
and asserts that the isomorphisms from the nerve theorem between
$H_p(|\setP|)$ and $H_p(\nerve(\setP))$ commute with inclusions.

\begin{theorem}[Chazal and Oudot]
\label{thm:nerve_inclusion_commute}
Let $\setP:=\{P^1,\ldots,P^n\}$ and $\setQ:=\{Q^1,\ldots,Q^n\}$ be good covers
with $P^i\subseteq Q^i$ for all $i=1,\ldots,n$.
Then, the isomorphisms $\phi_{\setP}: H_p(|\setP|)\rightarrow H_p(\nerve(\setP))$
and $\phi_{\setQ}: H_p(|\setQ|)\rightarrow H_p(\nerve(\setQ))$ commute with the
maps $i^\ast: H_p(|\setP|)\rightarrow H_p(|\setQ|)$ and $j^\ast: H_p(\nerve(\setP))\rightarrow H_p(\nerve(\setQ))$
that are induced by canonical inclusions, that is, $j^\ast\circ\phi_{\setP}=\phi_{\setQ}\circ i^\ast$.
\end{theorem}

\paragraph{Barcodes of shapes}
We let $\distance(\cdot, \cdot)$ denote the Euclidean distance function.
For a point set $A\subset\R^d$ and $x\in\R^d$,
we set $\distance(x,A):=\min_{y\in A} \distance(x,y)$. Then, $\distance(\cdot,A):\R^d\rightarrow\R$ is called the \emph{distance function from $A$}
and $A_\alpha:=\{x\in\R^d\mid \distance(x,A)\leq\alpha\}$ is called the \emph{$\alpha$-offset} of $A$.
With $\setP$ as above, we write $\setP_\alpha:=\{P^1_\alpha,\ldots,P^n_\alpha\}$ for the collection $\alpha$-offsets of $\setP$.
In particular, $\setP_0=\setP$.
We call $(|P_\alpha|)_{\alpha\geq 0}$ the \emph{offset-filtration} of $\setP$.
We pose the question of how to compute the barcode of the offset filtration of convex objects efficiently.
See Figure~\ref{fig:barcode_example} for an illustration of these concepts.

We define the analogue of \v{C}ech filtrations:
We call $(\nerve(\setP_\alpha))_{\alpha\geq 0}$ the \emph{nerve filtration}
of $\setP$;
it is indeed a filtration because for $\alpha_1\leq \alpha_2$, 
$\nerve(\setP_{\alpha_1})\subseteq\nerve(\setP_{\alpha_2})$.

\begin{theorem}
\label{barcode_equivalence}
Let $\setP$ be a collection of convex objects. Then, the $p$-barcodes of offset filtration and nerve filtration of $\setP$
are equal for all $p\geq 0$.
\end{theorem}
\begin{proof}
The nerve theorem yields an isomorphism of the homology groups for any parameter $\alpha$
and Theorem~\ref{thm:nerve_inclusion_commute} asserts that these isomorphisms commute with
inclusion. Using the \emph{persistence equivalence theorem}~\cite[p.159]{eh-computational},
the barcodes are equal.
\end{proof}

The nerve only changes for values where a collection of individual polyhedron offsets  
becomes intersecting. We call such an offset value \emph{nerve-critical}.
Since $P\subset\R^d$, we can restrict to collections of size at most $d+1$
since the $p$-barcode is known to be trivial 
for $p\geq d$ (since the $p$-th homology group is trivial for all $\alpha$).
Sorting the nerve-critical values $0=\alpha_0<\alpha_1<\ldots<\alpha_m$
and setting $K_i:=\nerve(\setP_{\alpha_i})$, the nerve filtration simplifies to the finite filtration
$K_0\subset K_1\subset\ldots\subset K_m$
whose barcode can be computed using standard methods; see~\cite{elz-topological,zc-computing} or~\cite{bkr-clear} for
an optimized variant.
Clearly, since $K_m$ contains a simplex for any subset of $\setP$ of size at most $d+1$,
its size is $\Theta(n^{d+1})$.

\section{Restricted barcodes}
\label{sec:restricted_barcodes}
Let $\setP:=\{P^1,\ldots,P^n\}$ be convex polyhedra in $\R^d$, that is, each $P^i$ is the intersection of finitely many half-spaces.
The major disadvantage of the construction of Section~\ref{sec:barcodes_of_shapes}
is the sheer size of the resulting filtration, $\Theta(n^{d+1})$. 
Our goal is to come up with a filtration that yields the same barcode
and is substantially smaller in size. Our approach is reminiscent
of alpha-complexes for point sets, but it requires additional ideas for being applicable
to convex objects.

From now on, we make the following assumptions for simplicity:
We refer to the elements of $\setP$ as \emph{sites}.
We restrict our attention to $d\in\{2,3\}$, that is, sites are polygons ($d=2$)
or polyhedra ($d=3$).
We assume the sites to be pairwise disjoint and in \emph{general position}, that is,
for any pair $P^i, P^j$ of sites, there is a unique pair of points $x^i\in\partial P^i$, $x^j\in\partial P^j$
that realizes the distance between the sites.
Moreover, we assume that the number of vertices, edges and faces of each site is bounded by a constant.
For a point $p\in\R^d$,
the site $P^k$ is \emph{closest} if $\distance(p,P^k)\leq \distance(p,P^\ell)$ for any $1\leq\ell\leq n$.
We assume for simplicity the generic case that no point has more than $d+1$ closest sites.
As a notational shortcut, we will frequently write $\distance(x):=\distance(x,|\setP|)$.

The \emph{Voronoi diagram} $\voronoi(\setP)$ is the partition of the space into maximal 
connected components with the same set of closest sites.
The Voronoi diagram is an \emph{arrangement} in $\R^d$, and its \emph{combinatorial complexity}
is the number of cells.
The \emph{Voronoi region} of $P^k$, denoted by $V^k$, is the (closed) set of points for which $P^k$ 
is one of its closest sites.
For a cell $\sigma$ of $\voronoi(\setP)$, we call $\critical{\sigma}:=\inf_{x\in \sigma} \distance(x)$ the \emph{critical value} of $\sigma$
(recall that $\distance(x)=\distance(x,|\setP|)$)
and a point $x$ that attains this infimum a \emph{critical point} of $\sigma$. Note that critical points of a cell may lie on its boundary.

For any two sites $P^i$, $P^j$, the \emph{bisector} $B$ is the set of points $x$ that satisfy $\distance(x,P^i)=\distance(x,P^j)$.
By general position of the sites, there is a unique point on $B$ that minimizes $\distance(\cdot,P^i)$. 
More generally, for $\alpha\geq 0$, let $B_\alpha:=\{x\in B\mid \distance(x,P^i)\leq \alpha\}$. We will frequently use the fact
that for any $\alpha$, $B_\alpha$ is empty or contractible.
This is implied by the following statement, which generalizes the well-known
\emph{pseudodisk-property}~\cite{klps-union}, \cite[Thm.13.8]{dutch}.
We will prove the case $d=3$ in Appendix~\ref{app:pseudodisk}.

\begin{theorem}
\label{thm:pseudodisk_maintext}
For $d\in\{2,3\}$, let $P_1$, $P_2$ be two convex disjoint polytopes in $\R^d$ 
in general position and let $B$ be the unit ball. 
Then, $\partial(P_1\oplus B) \cap \partial(P_2\oplus B)$
is either empty, a single point, or homeomorphic to $(d-2)$-sphere.
\end{theorem}

The \emph{restricted $\alpha$-offset} of $P^k$ is defined as $Q^k_\alpha:=P^k_\alpha\cap V^k$.
We set $\setQ_\alpha:=\{Q^1_\alpha,\ldots,Q^n_\alpha\}$ and $\setQ:=\setQ_0$.
In the same way as in Section~\ref{sec:barcodes_of_shapes}, 
we define the \emph{restricted nerve filtration} as $(\nerve(\setQ_\alpha))_{\alpha\geq 0}$
and \emph{$\setQ$-critical} values as those values where a simplex enters the restricted
nerve filtration.
The restricted nerve filtration can be expressed by a finite sequence of simplicial complexes 
that changes precisely at the $\setQ$-critical values.
The size of the filtration is bounded by the combinatorial complexity of the Voronoi diagram.
Moreover, the $\setQ$-critical value of a simplex associated with a Voronoi cell $\sigma$
equals the critical value of $\sigma$.

\begin{wrapfigure}[6]{r}{2.4cm}
\vspace{-.3cm}
\includegraphics[width=2.4cm]{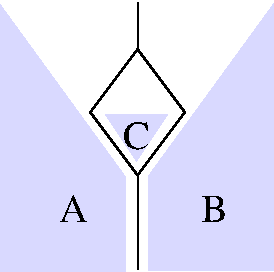}
\end{wrapfigure}
Restricting the offsets to Voronoi regions brings a problem: 
$\setQ_\alpha$ is not necessarily a collection of convex sets,
since $V^k$ is not convex in general. 
Even worse, $\setQ_\alpha$ might not be a good cover.
For instance, on the right we see three sites
$A,B,C$ and the induced Voronoi diagram (in black).
We see that the Voronoi regions of $A$ and $B$ intersect in two segments.
This means that the proof strategy of Theorem~\ref{barcode_equivalence} breaks down
since the nerve theorem does not apply.

\section{Restricted barcodes in 2D}
\label{sec:restricted_barcodes_2d}
We first restrict to the case $d=2$, that means, our input sites are interior-disjoint convex
polygons in the plane. While the restriction of offsets invalidates the proof of 
Theorem~\ref{barcode_equivalence}, it does not invalidate the statement, at least in
dimensions $0$ and $1$.

\begin{theorem}
\label{thm:1-barcode_equivalence}
For convex polygonal sites in $\R^2$, the $0$- and $1$-barcode of the restricted nerve filtration are equal to the $0$- and $1$-barcode
of the offset filtration, respectively.
\end{theorem}

As a consequence of this theorem, we obtain a filtration of size $O(n)$ that has the same
barcode as the offset filtration; the size follows from the fact that the complexity
of the Voronoi diagram is $O(n)$.
This is much smaller than the $O(n^3)$ filtration obtained by the unrestricted nerve.
The construction time is dominated by computing the Voronoi diagram and thus bounded by $O(n\log n)$~\cite{yap-vd}.

\begin{wrapfigure}[10]{r}{2.5cm}
\vspace{-.5cm}
\includegraphics[width=2.5cm]{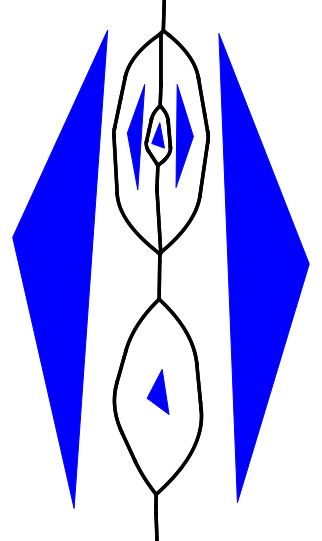}
\end{wrapfigure}
We provide a sketch of the proof of Theorem~\ref{thm:1-barcode_equivalence}, 
and refer the reader to Appendix~\ref{app:monster_proof} for the complete proof. 
We concentrate on the $1$-barcode for the proof; the (simpler) statement for the $0$-barcode follows with similar arguments.
We first analyze more carefully what causes two restricted offsets to have a non-contractible intersection. 
Such a non-contractible intersection only happens when two restricted offsets intersect in several
components, and thus the corresponding two sites contribute two or more bisector segments to the Voronoi diagram. 
This, in turn, only happens if their bisector is \emph{split} by another site that ``sits in-between''.
Precisely, observe that the intersection of two restricted offsets consists of $k$ connected components (with $k\geq 1$)
if and only if the complement of their union induces $k-1$ bounded regions. We call these bounded regions 
\emph{surrounded region} induced by the restricted offsets of two sites.
A surrounded region contains at least one and potentially more sites and can therefore contain nested surrounded regions
(induced by two restricted offsets within the surrounded regions). 
We call a surrounded region \emph{simple} if it does not contain any other surrounded region.
On the right, we see an example where the left and right
sites induce two surrounded regions (for $\alpha$ large enough), the upper not being simple because it contains a nested surrounded region. 

Fix two sites $A$ and $B$ and assume that their restricted $\alpha$-offsets leave some surrounded region $R$ for some fixed $\alpha$.
The crucial observation is that whenever this happens, $R$ is always already ``filled'':

\begin{lemma}
The (unrestricted) $\alpha$-offsets of $A$ and $B$ contain $R$. In particular, the restricted $\alpha$-offsets of sites in $R$ fill out the entire region $R$.
\end{lemma}
\begin{proof}
Let $a$ and $b$ denote the boundary curves of $A$ and $B$. Let $v_1$, $v_2$ denote the points
on the boundary of $R$ that lie in $a\cap b$. 
Assume wlog that $\distance(v_1)\leq \distance(v_2)=:w\leq\alpha$. 
The bisector of $A$ and $B$ has a segment within $R$ that connects $v_1$ and $v_2$. 
Since sublevel sets on the bisector are connected, we have that $\distance(x)\leq w$ for all $x$ on that bisector segment.
Moreover, for any $x$ on the part of $a\setminus b$ that bounds $R$, we must have $\distance(x)\leq w$ as well. Combining these two properties,
the ``half-region'' of $R$ bounded by $a\setminus b$ and the bisector segment satisfies $\distance(x)\leq w$ on its boundary and by convexity
of the distance function, $\distance(x)\leq w\leq\alpha$ in the whole region. Applying the same argument on the other half-region, we get the result.
\end{proof}

For a fixed $\alpha$, we call a site \emph{surrounded} if it lies in some surrounded region, and \emph{free} otherwise. Recall that we write
$\setQ_\alpha$ for the restricted $\alpha$-offsets of the sites, and we let $\setQ_\alpha^\ast\subseteq\setQ_\alpha$ denote the restricted $\alpha$-offsets
of the \emph{free} sites. The previous lemma implies that disregarding the surrounded sites does not change the offset, 
so $|\setQ_\alpha|=|\setQ_\alpha^\ast|$. Moreover, the free sites form a good cover because all surrounded regions have been removed by construction.
It follows that the nerve theorem applies and $H_P(|\setQ_\alpha|)=H_p(|\setQ_\alpha^\ast|)=H_p(\nerve(\setQ_\alpha^\ast))$ for all $p\geq 0$.

The first major technical result is that $H_1(\nerve(\setQ_\alpha^\ast))=H_1(\nerve(\setQ_\alpha))$. Combined with the previous statement, this implies
that for any $\alpha$, the first homology group of the restricted $\alpha$-offsets is isomorphic to the first homology group of its nerve.
The construction of the isomorphism is iterative, always removing the sites in an innermost, and thus simple, surrounded region at a time. Let us denote by $S$ the set of sites that are not removed yet. Initially, $S$ is the set of all sites
and at the end, $S$ is the set of free sites. Let us fix a simple surrounded region $R$, and let $A$, $B$ be the sites surrounding it.
Let $M_R$ be the sites within $R$ and let $S_R:=S\setminus M_R$. We define a map from $S$ to $S_R$ that maps all sites in $M_R$ to $A$, and each remaining site to itself.
This map assigns vertices of $\nerve(S)$ to vertices of $\nerve(S_R)$, and induces a simplicial map $\phi$ between the nerves because $\{A,B\}$ separates $M_R$ from all other sites in the nerve.
Being a simplicial map, $\phi$ induces a map $\phi^\ast:H_1(\nerve(S))\rightarrow H_1(\nerve(S_R))$ of homology groups. 

\begin{lemma}
$\phi^\ast$ is an isomorphism.
\end{lemma}
\begin{proof}
We set $M_R^{\mathrm{ext}}:=M_R\cup\{A,B\}$ and argue that $H_1(\nerve(M_R^{\mathrm{ext}}))=0$: We can restrict $A$ and $B$ to a neighborhood around $R$ without changing the nerve.
Since $R$ is simple, $M_R^{\mathrm{ext}}$ is ``almost'' a good cover~-- the only obstacle is that $A$ and $B$ intersect in two components. However, we can cut one of the two
intersections open. This will cause at most one triangle in the nerve to be removed
(losing this triangle is the reason why the statement does not extend to $H_2$). 
However, since $R$ is filled, the underlying space is a disk even after the cut, and the nerve theorem asserts that $H_1(\nerve(M_R^{\mathrm{ext}}))=0$.

The proof works now by fixing a cycle in $\nerve(S)$ and transforming it to a cycle that does not contain any element of $M_R$ anymore.
The idea is to ``reroute'' any path that enters $R$ such that it runs entirely in $A\cup B$. Because the nerve of $M_R\cup\{A,B\}$ has trivial $1$-homology,
such a rerouting does not change the homology of the cycle. We skip further details of this elementary construction.
\end{proof}

It follows that our construction yields a sequence of isomorphisms connecting $H_1(\nerve(\setQ_\alpha^\ast))$ and $H_1(\nerve(\setQ_\alpha))$, thus proving the equivalence
of the $1$-homology for every $\alpha$. This statement, however, does \emph{not} immediately imply the equivalence of the $1$-barcodes of both, 
because it is not clear that the constructed isomorphisms commute with inclusion. 
This is the second major technical result needed for the proof.
With $\alpha_1\leq\alpha_2$, we have the diagram

\begin{eqnarray}
\xymatrix{
H_1(|\setQ_{\alpha_1}|) \ar@{=}[r]\ar @{^{(}->}[d] & H_1(|\setQ_{\alpha_1}^\ast|)\ar @{^{(}->}[d]  & H_1(\nerve \setQ_{\alpha_1}^\ast)\ar[l]^{\theta^\ast}\ar[d]^{\phi^\ast} & H_1(\nerve \setQ_{\alpha_1})\ar @{^{(}->}[d]  \ar [l]^{\phi^\ast}\\
H_1(|\setQ_{\alpha_2}|) \ar@{=}[r] & H_1(|\setQ_{\alpha_2}^\ast|)  & H_1(\nerve \setQ_{\alpha_2}^\ast)\ar[l]^{\theta^\ast} & H_1(\nerve \setQ_{\alpha_2})  \ar [l]^{\phi^\ast}\\
}
\end{eqnarray}
where $\theta^\ast$ is the isomorphism from the nerve theorem and $\phi^\ast$ a composition of the isomorphisms between the nerves as constructed above.
The first and the third square commute as one can easily verify. The difficulty lies in the the middle square.
Note that Theorem~\ref{thm:nerve_inclusion_commute} does not apply here because the nerve filtration of free sites is not induced
by inclusion. Indeed, when $\alpha$ increases, sites may change their status from free to surrounded, and thus disappear
from the set $\setQ^\ast$. However, the map $\phi$ as defined above induces a natural simplicial map,
so the nerves of free sites form a filtration induced by simplicial maps. The remainder of the proof consists of investigating the
internals of the nerve isomorphism, similar to the proof of~\cite[Lemma 3.4]{co-towards}, and to show that it commutes
with the simplicial map $\phi$ (on the level of chain groups). We skip further details.
With the commutativity of the diagram, Theorem~\ref{thm:1-barcode_equivalence}
follows from the persistence equivalence theorem~\cite[p.159]{eh-computational}.

\begin{wrapfigure}[9]{r}{2.5cm}
\vspace{-.5cm}
\includegraphics[width=2.5cm]{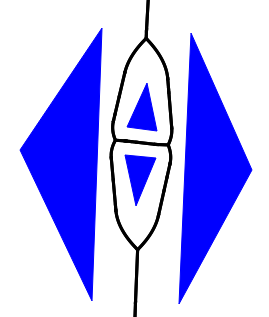}
\caption{A ghost sphere}
\label{fig:ghost_sphere}
\end{wrapfigure}
Theorem~\ref{thm:1-barcode_equivalence} does not generalize to the $2$-barcode.
For instance, in Figure~\ref{fig:ghost_sphere}, 
we see four sites in $\R^2$ where every triple of Voronoi regions intersects, but there is no common intersection
of all four of them. Consequently, their nerve consists of the four boundary triangles of a tetrahedron
and therefore carries non-trivial $2$-homology. We refer to such homology classes as ``ghost features''.
In the planar case, the offset filtration can clearly not form any void (a 
$3$-dimensional hole) and we can therefore safely ignore all ghosts.
In $\R^3$, however, the $2$-barcode carries information about the offset and the ghosts need to be distinguished
from real features. 
At first glance, one might hope that ghost features have infinite persistence
(as opposed to real features). However, we can quite easily
extend the situation of Figure~\ref{fig:ghost_sphere}
to four prisms in $\R^3$ that create a ghost with finite persistence
This shows that considering only the nerve is problematic in dimensions
higher than $2$.

\section{Restricted barcodes in 3D}
\label{sec:higher_dim}
As Theorem~\ref{thm:1-barcode_equivalence} does not generalize to higher-dimensions,
we now present a refinement of the nerve construction for three-dimensional space.
Reconsidering the ``ghost example'' from Section~\ref{sec:restricted_barcodes},
it seems attractive to pass to the multi-nerve~\cite{multinerve}, that is, introducing a distinct
simplex for each lower-dimensional cell of the Voronoi diagram.
However, this approach is not sufficient for two reasons:
\begin{inparaenum}
\item Voronoi cells might be non-simply connected;
\item even if the Voronoi cells form a good cover, this may not be true for the restricted offsets at all scales $\alpha$.
\end{inparaenum}
See Figure~\ref{fig:not_simple_3_d} for an illustration.

\begin{figure} %
	\centering
	\includegraphics[width=0.25\textwidth] {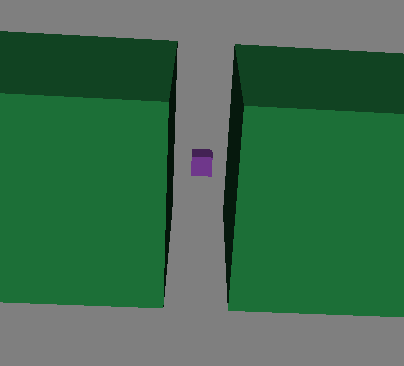}
	\hspace{1cm}
	\includegraphics[width=0.25\textwidth] {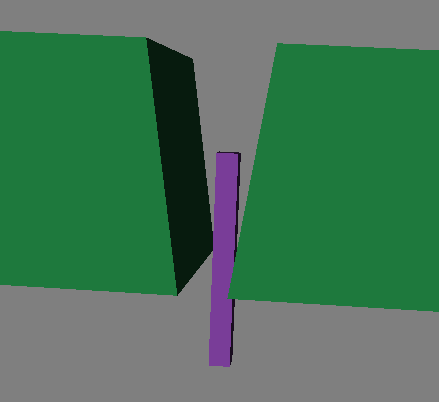}
	
	\caption{
	Left: Two large polyhedra with a small cube in between 
	them. The Voronoi cell of the two large polyhedra will contain an unbounded face with a hole 
	in its middle. 
	Right: Two large polyhedra that are closer at their bottom than at their top. 
	At the time when their restricted offsets first intersect, the intersection
	will be composed of two connected components at their bottom, one at each side of the purple 
	polyhedron, despite the fact that the Voronoi cells form a good cover.
	This and the remaining figures are best viewed in color.
	}
	\vspace{-3mm}
	\label{fig:not_simple_3_d}
\end{figure}

We use the following definitions: 
an arrangement $\mathcal{A}$ in $\R^3$ is a \emph{refinement} 
of $\voronoi(\setP)$ in $\R^3$
if every $0$-, $1$-, $2$-, or $3$-dimensional cell of $\mathcal{A}$ is contained in a cell of $\voronoi(\setP)$.
For a cell $\sigma\in\mathcal{A}$ and $\alpha\geq 0$, define the \emph{restricted cell} $\sigma_\alpha:=\{x\in\sigma\mid \distance(x)\leq \alpha\}$.
We call $\sigma$ \emph{stratified} if for all $\alpha\geq 0$, $\sigma_\alpha$ is empty or contractible. Note that in particular, 
a stratified cell is contractible.
We call an arrangement $\mathcal{A}$ a \emph{stratified refinement} of $\voronoi(\setP)$,
if $\mathcal{A}$ is a refinement of $\voronoi(\setP)$ and every cell of $\mathcal{A}$ is stratified.
As before, we define the critical value of a cell $\sigma\in\mathcal{A}$ as $\critical{\sigma}:=\inf\{\alpha\in\R\mid \sigma_\alpha\neq\emptyset\}$.

\paragraph{Co-arrangements}
An arrangement $\mathcal{A}$ in $\R^3$ gives rise to a dual structure in a natural way:
fixing two cells $\sigma,\tau$ of $\mathcal{A}$ with $\dim(\sigma)<\dim(\tau)$, 
we have that either $\sigma$ is contained in or completely disjoint from the boundary of $\tau$. 
In the former case, we say that $\sigma$ is \emph{incident} to $\tau$. 
If $\sigma$ is incident to $\tau$, $\critical{\sigma}\geq\critical{\tau}$.
The $\emph{co-arrangement}$ $\mathcal{A}^\ast$ of $\mathcal{A}$ is defined as follows:
for every cell $\sigma$ of $\mathcal{A}$, $\mathcal{A}^\ast$ has a \emph{co-cell} $\sigma^\ast$ such that their dimensions add up to $3$. The \emph{boundary} of a co-cell $\sigma^\ast$ of dimension $\delta$,
$\partial(\sigma^\ast)$, is the set of all co-cells $\tau^\ast$ of dimension $\delta-1$ such that $\sigma$ is incident to $\tau$. See Figure~\ref{fig:thickening} (left) for an illustration of the corresponding concept in the planar case.

\begin{figure}[t]
\centering
\includegraphics[width=11cm]{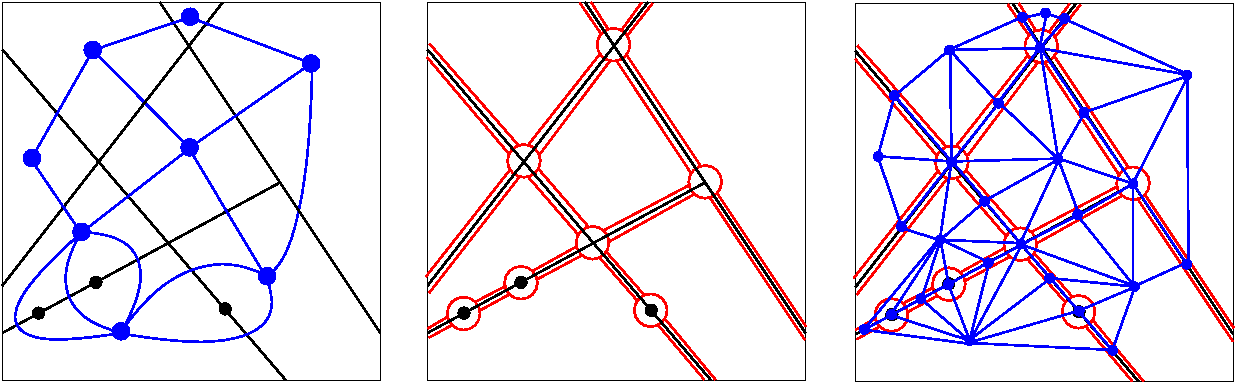}
\caption{
Left: An arrangement $\mathcal{A}$ (black) and the co-arrangement (blue). Middle: The thickening of $\mathcal{A}$. Right: The nerve of the thickening.}
\vspace{-3mm}
\label{fig:thickening}
\end{figure}

The critical value of a co-cell is defined as the critical value of its primal counterpart. 
This turns the co-arrangement into a filtered cell complex, since any co-cell has a critical value not smaller than
any co-cell in its boundary. For $\alpha\in[0,\infty]$, we let $\mathcal{A}^\ast_\alpha$ denote the collection of co-cells with critical value at most $\alpha$. 
Since $\partial\partial(\sigma)=0$ for any $\sigma$, there is a well-defined homology group for each $\mathcal{A}^\ast_\alpha$, and therefore, a barcode of the co-arrangement $\mathcal{A}^\ast$%
(equivalently, we could apply the \emph{cohomology functor} on the \emph{cofiltration} $\mathcal{A}_\alpha$).
We can now state the main result which allows us to express the barcode of the offset of three-dimensional 
shapes as well, in terms of a combinatorial structure.

\begin{theorem}
\label{thm:coarrangement}
Let $\setP$ be a collection of convex polyhedra in $\R^3$, and
let $\mathcal{A}$ be a stratified refinement of $\voronoi(\setP)$. Then, the barcode of the offset filtration of $\setP$
equals the barcode of $\mathcal{A}^\ast$.
\end{theorem}

\begin{proof}
For each cell $\sigma\in\mathcal{A}$, 
we define a \emph{thickening of $\sigma$} by transforming $\sigma$ into a $3$-dimensional cell.
The construction yields a collection $\thick$ of thickenings which are closed and interior-disjoint.
The construction idea is to define the thickening of $\sigma$ by an offset of $\sigma$
with a sufficiently small radius. More precisely, we choose radii $\eps_0>\eps_1>\eps_2$ and proceed in increasing dimensions.
The thickening of an $i$-cell $\sigma$ is its $\eps_i$-offsets, restricted to the space that has not been occupied yet
by thickenings of lower-dimensional cells. Choosing $\eps_0$ small enough and $\eps_0\gg\eps_1\gg\eps_2$, we can ensure
that two thickenings intersect if and only if the corresponding two cells are incident.
More generally, a set of thickenings has a non-empty common intersection if and only the involved cells $\sigma_1,\ldots,\sigma_k$
form a \emph{flag}, that is, $\sigma_i$ is incident to $\sigma_{i+1}$ for $1\leq i\leq k-1$.
Finally, for sufficiently small $\eps_i$, the intersection of such a thickened flag has a homotopy-preserving 
deformation to $\sigma_1$, the lowest dimensional cell of the flag. 
Since all cells are contractible, this implies that the thickenings form a good cover.
We define $\thick$ as the collection of the thickening of all cells in $\mathcal{A}$.
See Figure~\ref{fig:thickening} (middle) for an illustration.
For any $\alpha\geq 0$, we define the $\alpha$-thickening $\thick_\alpha\subseteq\thick$ as the subset consisting of all 
thickenings of cells of $\mathcal{A}$ with a critical value of at most $\alpha$.
$\thick_\alpha$ forms a good cover as well just because it is a subset of the good cover $\thick$.
Recall that $|\thick_\alpha|$ denotes the underlying space of~$\thick_\alpha$.

The first part of the proof is to show that the filtration $(|\thick_\alpha|)_{\alpha\geq 0}$ has the same barcode as the offset filtration $(|\setP_\alpha|)_{\alpha\geq 0}$.
For that, we define a continuous transformation from $|\setP_\alpha|$ to $|\thick_\alpha|$ through a sequence of expansions and retractions.
This transformation ensures that the spaces are homotopically equivalent 
(more formally the intersection is a \emph{deformation retract} for both spaces). 
To construct the deformation, 
we let $X$ denote the deformed shape; initially, $X$ equals $|\setP_\alpha|$.
We proceed in two phases, 
a \emph{retraction phase} and an \emph{expansion phase}. For a fixed $\alpha$, we call a cell $\sigma$ \emph{active} if $\critical{\sigma}\leq\alpha$,
and \emph{inactive} otherwise.

For the retraction phase, let $\sigma$ be an inactive cell. 
Being inactive, the $\alpha$-offset did not reach a point on $\sigma$ yet~-- however, $X$ may well intersect the thickening of $\sigma$ already, 
because the offset of several sites are coming very close to $\sigma$ (with $c=\critical{\sigma}$ and $\delta=\dim(\sigma)$,
such an event occurs only in the small range $[c-\eps_\delta,c]$). The whole purpose of the retraction phase is to remove such features from $X$,
such that $X$ has no points inside the thickenings of inactive cells empty. 
The retraction works as follows: let $\gamma$ be a threshold value, starting at $\alpha$ and continuously decreasing towards zero.
The restriction of the distance function $\distance$ to an inactive cell $\sigma$ induces decreasing sublevel sets with respect to threshold $\gamma$.
This defines a retraction path for every point which has to hit the boundary, where it either stops (if we reach the boundary of an active cell)
or further retracts (otherwise). Since this operation does not affect any point on any cell $\sigma$ (active or inactive),
no tearing is happening, thus the transformation is a deformation retract.

The expansion works similarly in the other direction. We let $\gamma$ progress from $\alpha$ towards $\infty$. We let $X$ grow
within each active cell simultaneously by increasing the sublevel set threshold with respect to $\gamma$
(in informal terms, for active cells,
$X$ antedates the changes happening in the cell in the future and performs them at scale $\alpha$ all at once).
Unlike in the retraction phase, this operation affects points on cells $\sigma$ and we have to ensure that no unwanted gluing occurs. 
However, since $\sigma$ is active, there is at least one point on $\sigma$ in $X$ initially, for threshold value $\alpha$. Since by assumption,
the sublevel set $\sigma_\gamma$ stays contractible, we are only expanding existing intersection patterns, which is enough to argue
that the deformation preserves the homotopy type. 
At the end of the expansion, $X$ is precisely the union of all active cells, thus equal to $|\thick_\alpha|$. Also, the constructed deformations commute with inclusion on a homotopy level. 
Consequently, the two filtrations have the same barcode, which concludes the first part of the 
proof.

By duality, the flags of $\mathcal{A}_\alpha$ equal the flags of $\mathcal{A}^\ast_\alpha$ in the sense that any flag of $\mathcal{A}_\alpha$ is in one-to-one correspondence
to a sequence of co-cells $\sigma^\ast_1,\ldots,\sigma^\ast_k$ such that $\sigma^\ast_{i+1}$ is a face of $\sigma^\ast_{i}$ for $1\leq i\leq k$.
That implies that the nerve of the $\alpha$-thickenings equals the (abstract) barycentric subdivision of $\mathcal{A}^\ast_\alpha$. 
Since the $\alpha$-thickenings are a good cover, the filtrations $(|\thick_\alpha|)_{\alpha\geq 0}$
and $(\nerve(\thick_\alpha))_{\alpha\geq 0}$ have the same barcode. Putting everything together,
the barcode of the offsets has the same barcode as the barycentric subdivision filtration of $\mathcal{A}^\ast$.

To finish the proof, it suffices to argue that the barcode of $\mathcal{A}^\ast$ and its barycentric subdivision coincide. 
Note $\mathcal{A}^\ast$ has a natural geometric realization in $\R^3$ as the dual of an arrangement $\mathcal{A}$ in $\R^3$. 
Moreover, its barycentric subdivision permits the same geometric realization, by picking one point in the interior of each cell
and subdividing according to incidence relations. 
We can furthermore note that both the arrangements and the co-arrangements are CW-complexes. 
Since cellular and singular homology
coincide for this class of complexes, the two complexes have the same homology. 
The argument also holds for any $\alpha$ and the underlying map commutes with inclusions. This shows
the equality of the barcodes and proves the theorem.
\end{proof}

\paragraph{Stratified refinements}
We are left with the question of how to obtain a stratified refinement of a Voronoi diagram of convex polyhedra.
We remark that for the case $d=2$, the Voronoi diagram of convex polygons is already stratified,
so no refinement is needed. This follows directly from the pseudodisk property (Theorem~\ref{thm:pseudodisk_maintext}),
ensuring that the sublevel set of each bisector stays connected for each $\alpha$.
The resulting co-arrangement is precisely the multi-nerve as defined in~\cite{multinerve}.
However, we note that in light of Theorem~\ref{thm:1-barcode_equivalence},
there is no need to consider co-arrangements at all in the planar case.

We turn to the case $d=3$. Here, the Voronoi diagram is generally not stratified, for the reasons given at the beginning 
of the section. To stratify, we will cut every cell into stratified pieces. We consider the case of a single connected component 
of a bisector (that is, a $2$-cell of the Voronoi diagram) in isolation: as it will turn out, our stratification method will 
ensure that all trisectors (i.e., points with the same distance to three sites) on its boundary will become stratified as well. 

So, let us fix a $2$-cell $\sigma$, contained in the bisector $B$ of two sites $P^i, P^j\in\setP$.
The boundary of $\sigma$ consists of a collection of $1$-cells where each $1$-cell belongs to a trisector.
The boundary splits into connected components;
since $B$ is homeomorphic to a plane, we can distinguish between an \emph{outer boundary} of $\sigma$
(which might be unbounded if $\sigma$ is unbounded)
and an arbitrary number of closed \emph{inner boundaries}.
The presence of inner boundaries turns $\sigma$ non-simply connected, and thus $\sigma$ is non-stratified
(an example is given by Figure~\ref{fig:not_simple_3_d} (left)).

Note that for $x\in\sigma$, we have that $\distance(x)=\distance(x,|\setP|)=\distance(x,P^i)=\distance(x,P^j)$ by definition.
Therefore, $\distance$ cannot have more than one local minimum \emph{in the interior} of $\sigma$ since two such minima
would imply two minima on the bisector, which is impossible by Theorem~\ref{thm:pseudodisk_maintext}.
However, it is well possible that $\distance$ restricted to $\sigma$ has local minima on $\partial(\sigma)$,
both on inner and outer boundary components. The presence of such local minima turns the restricted cell $\sigma_\alpha$
disconnected for a certain range of scales and turns $\sigma$ non-stratified as well (Figure~\ref{fig:not_simple_3_d} (right)).

We now define a stratified refinement of $\sigma$ into $2$-cells by introducing cuts in $\sigma$. 
We start by cutting $\sigma$ along a curve $\psi$ on $B$ that goes through $o$ and that is unimodal for $\distance$,
that is, has a minimum at $o$ and $\distance$-monotone otherwise. Moreover, we require $\psi$
to intersect every $1$-cell in the boundary of $\sigma$ only a constant number of times. Such a curve indeed exists
since we can find two monotone paths from $o$ to the outer boundary that avoid all inner boundaries.

Having cut using $\psi$, we introduce further cuts as follows:
for any value $\beta>0$, the \emph{$\beta$-isoline} is the curve defined by all points $p$ on the bisector 
with $\distance(p) = \beta$. 
By Theorem~\ref{thm:pseudodisk_maintext}, a $\beta$-isoline is a closed cycle on the bisector 
that loops around $o$. Restricting $\distance(\cdot)$ to 
$\partial \sigma$, we get that $\distance(\cdot)$ attains local maxima and minima at points on the (inner and outer) 
boundary curves, which we refer to as critical points. For a critical point $p$ with $\beta=\distance(p)$, 
an \emph{isoline segment at $p$} is a maximal connected piece of the $\beta$-isoline within $\sigma$ with $p$ on its boundary. 
An isoline segment may degenerate into a point.
We introduce cuts along all isoline segments for all critical points. 
Since $\psi$ is unimodal, unbounded, and goes through $o$, it intersects any isoline twice. 
Hence, isoline segments cannot be closed curves. 
An illustration is shown in Figure~\ref{fig:stratified_refinement}.
Let $\hat{\sigma}\subseteq\sigma$ be any cell of the obtained refinement of $\sigma$.

\begin{figure} %
	\centering
	\subfloat[][The original face $\sigma$]{
		\includegraphics[width=0.3\textwidth] {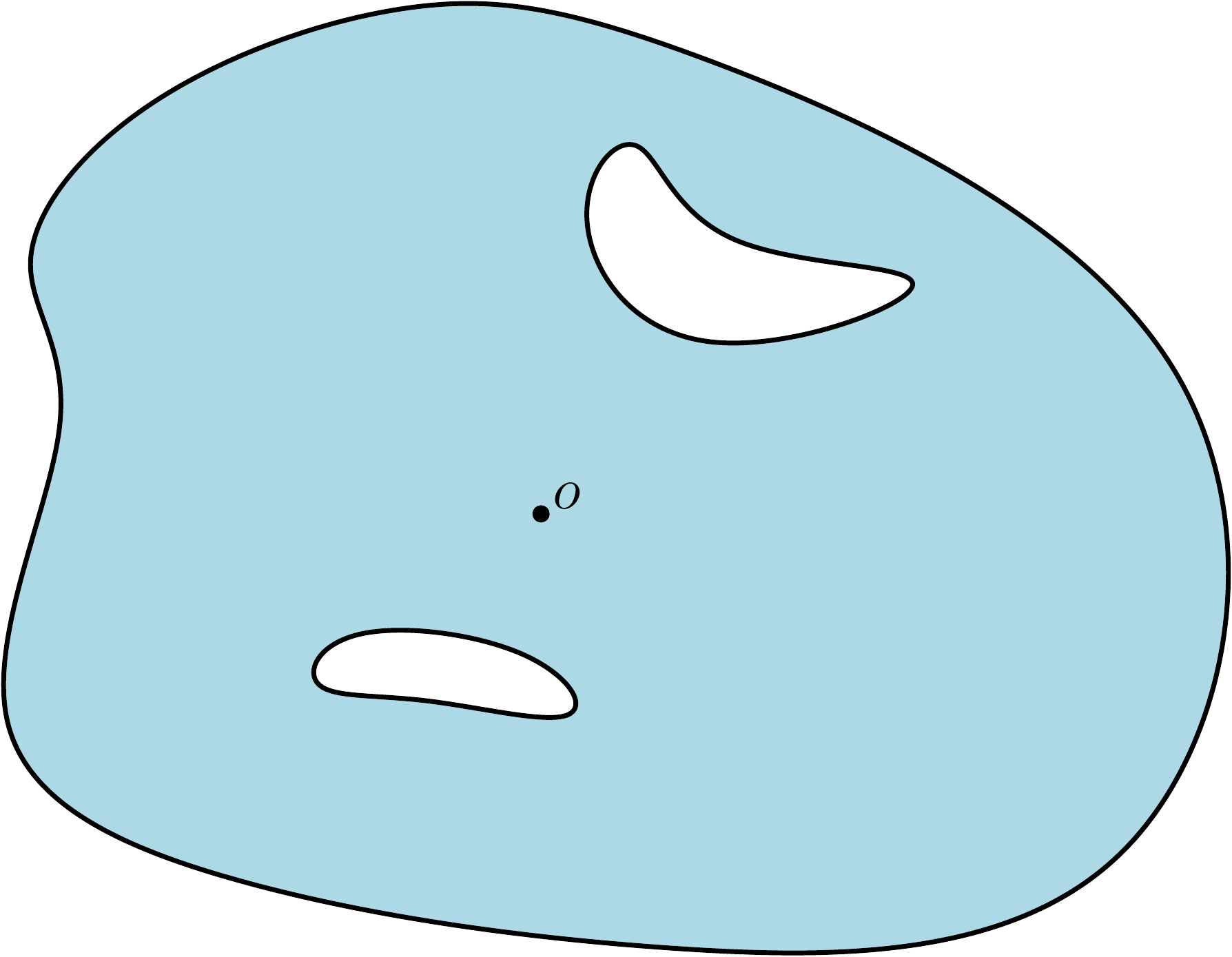}
		\label{subfig:stratified_original}
	}
	\hspace{0.3cm}
	\subfloat[][The isoline cuts over $\sigma$] {
		\includegraphics[width=0.3\textwidth]{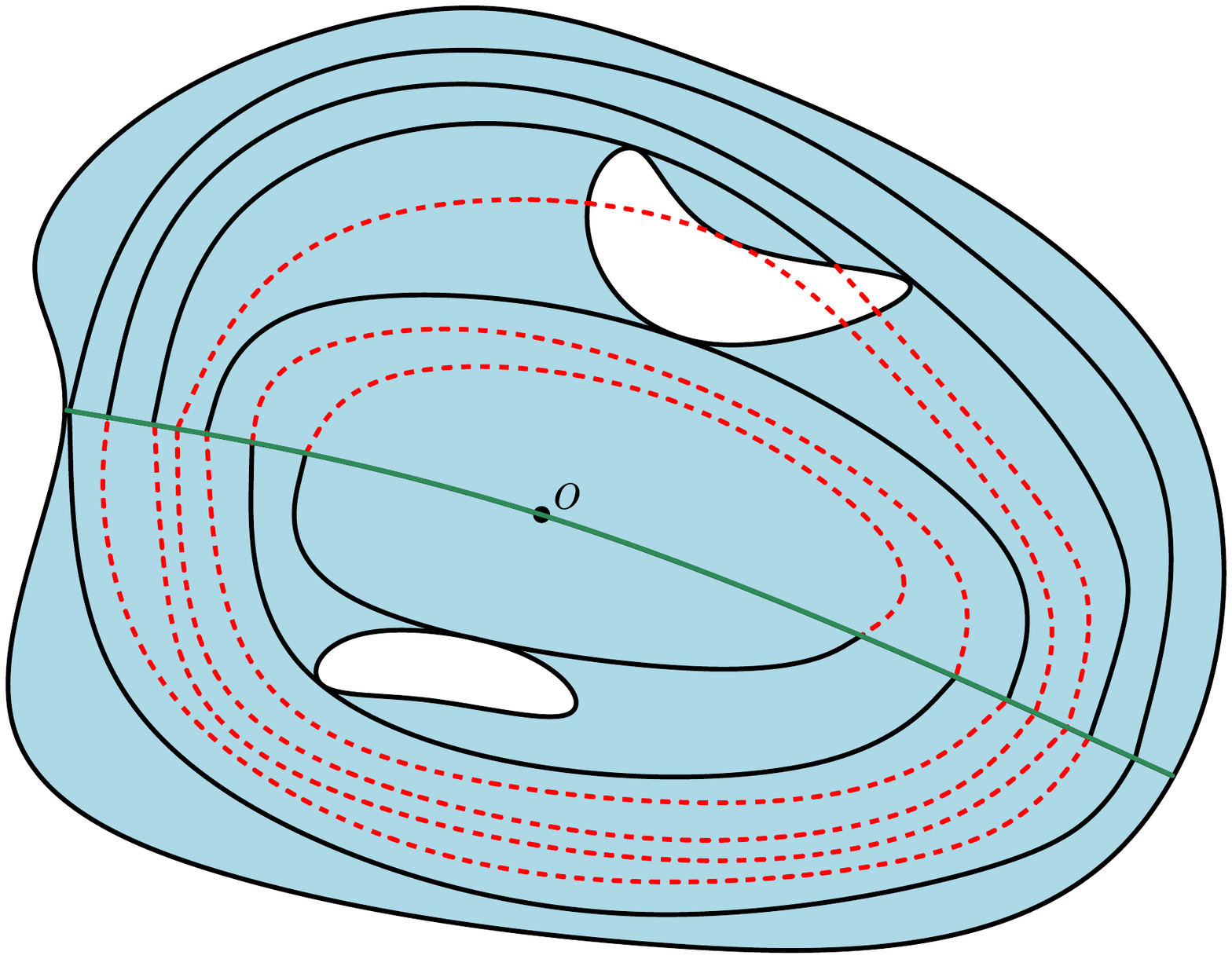}
		\label{subfig:stratified_with_isolines}
	}
	\hspace{0.3cm}
	\subfloat[][The stratified refinement] {
		\includegraphics[width=0.3\textwidth]{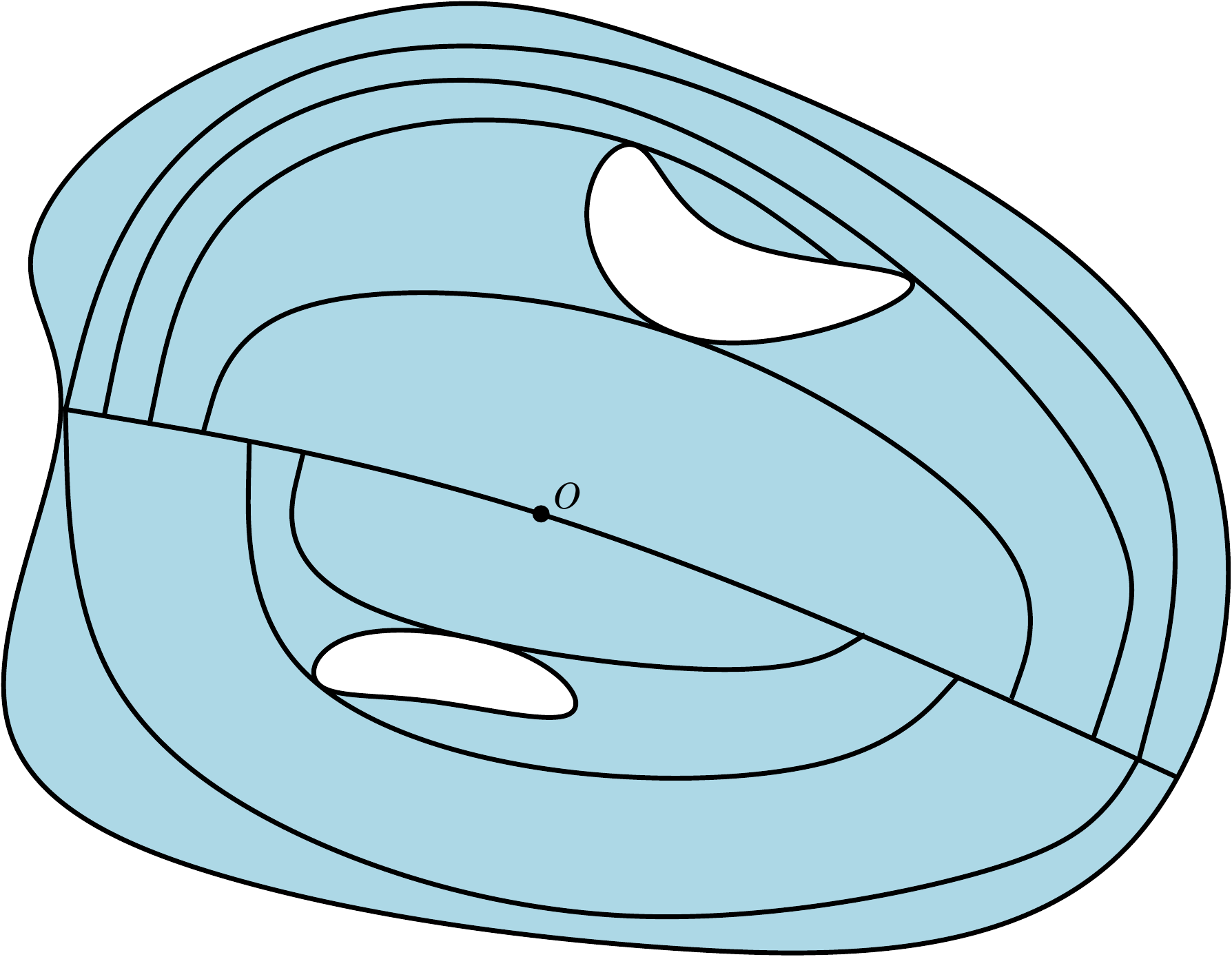}
		\label{subfig:stratified_final}
	}
	
	\caption{
	A stratified refinement of a $2$-cell. The original $2$-cell $\sigma$ is shown 
	in~$\protect\subref{subfig:stratified_original}$. In~$\protect\subref{subfig:stratified_with_isolines}$ we add the cut induced by 
	the unbounded curve $\psi$ (green) and cuts induced by the isoline segments at 
	critical points (black). The part of an isoline that is not an isoline segment is shown in red.
	The final stratification is shown in~$\protect\subref{subfig:stratified_final}$.
	}
	\label{fig:stratified_refinement}
	\vspace{-3mm}
\end{figure}

\begin{lemma}
Any sublevel set of a refined cell $\hat{\sigma}$ is contractible.
\end{lemma}

\begin{proof}
Fix a cell $\hat{\sigma}$ and consider $\distance$
restricted to $\hat{\sigma}$. $\hat{\sigma}$ cannot have any (local) minimum
in its interior, since such a local minimum would be a local minimum
of the bisector and thus must be equal to $o$. 
However, $\psi$ cuts through $o$, so $o$ is not in the interior of any cell.
Let $p\in\hat{\sigma}$ be a point with minimal $\distance$-value in $\hat{\sigma}$.
$p$ must lie on some (possibly degenerate) isoline segment $\gamma$
because otherwise, $p$ would lie either on a segment of $\psi$ 
(which is $\distance$-monotone)
or on the boundary of $\sigma$ 
(which we have cut into $\distance$-monotone pieces).
Following the boundary of $\hat{\sigma}$, the neighboring segments of $\gamma$
are two monotone curves $\xi_1$, $\xi_2$ (in fact, the neighbor
of $\gamma$ might also be the second isoline segment 
emanating from the same critical
in which case we consider $\gamma$ to be their union).

We sweep an isoline segment in a $\distance$-increasing direction which is anchored at points on $\xi_1$ and $\xi_2$,
starting with $\gamma$.
This traces out a subset of $\hat{\sigma}$ 
with the property that every sublevel set is contractible.
There are three cases possible: 
\begin{inparaenum}
\item $\xi_1$ and $\xi_2$ intersect
in which case $\hat{\sigma}$ is bounded by $\gamma$, $\xi_1$ and $\xi_2$.
\item the sweep isoline segment hits a point that is 
not in the interior of $\sigma$. Then, it follows that this intersection
point must lie on an isoline segment $\gamma'$ as well 
(it cannot lie on a monotone curve by minimality of the isovalue). 
By construction, $\gamma'$ connects $\xi_1$ and $\xi_2$
\item if none of the first two cases arise, $\hat{\sigma}$ is unbounded.
\end{inparaenum}
In all cases, we observe that the sweeping process traces out the entire
set $\hat{\sigma}$, which implies that every sublevel set is contractible.
\end{proof}

The cuts induce a subdivision of the $1$-cells as well. It is easy to see that no refined $1$-cell can be a cycle --
any closed curve is cut at its maxima and minima, and newly introduced isoline segments cannot not form a cycle.
Therefore, all $1$-cells in $\sigma$ are simply-connected as well. This implies that the construction yields a 
stratified refinement of the Voronoi diagram.

\paragraph{Size of the stratified refinements}

The described refinement only increases the complexity of the arrangement by a constant factor, 
as we shall now show.
Recall that a trisector is the curve of points
that have the same distance to three sites. The distance function restricted to a trisector
defines a function on this curve. In the refinement procedure, we consider the local extrema
of this curve. We first show the following statement.

\begin{lemma}
\label{lem:local_extrema}
The number of local extrema of a trisector is in $O(1)$.
\end{lemma}
\begin{proof}
A trisector is defined as the set of points of equal distance to three polyhedra. 
Each such point is of equal distance to a face, an edge or a vertex of each of the polyhedra. 
This implies that a trisector of polyhedra is composed of pieces of trisectors of points, lines, and planes that 
support the appropriate features of the polyhedra. 
Each such trisector can be defined by a set of algebraic equations of constant degree, using the squared distances 
from a point to a point, a line, or a plane in $\R^3$. 
Maximizing or minimizing the (squared) distance function on such a trisector can be phrased
as a Lagrange optimization problem with objective function and constraints being algebraic curves of constant
degree. In this case, the condition of being a maximum can again be expressed
as an algebraic equation of constant degree, defining a surface. Cutting this surface with the trisector
yields a constant number of points, assuming that the sites are in generic position.
\end{proof}

With that lemma, it is straight-forward to prove the size bound

\begin{theorem}
For a set $\setP$ of convex, disjoint polyhedra in generic position, 
let $\pi(n)$ be the complexity of $\voronoi(\setP)$.
Then, there exists a stratified refinement of $\voronoi(\setP)$ with $O(\pi(n))$ cells.
\end{theorem}

\begin{proof}
We charge the cuts of a cell $\sigma$ to $0$- and $1$-cells on its boundary.
We start with the cuts introduced by $\psi$: every such cut is between two $1$-cells, and we can
charge the cut to either of them (if $\psi$ cuts through a $0$-cell, we can charge the cut
to one of the incident $1$-cells on the boundary).
Since a $1$-cell can only participate in the boundary of three $2$-cells, and we assume
that $\psi$ hits any $1$-cell only a constant number of times, every $1$-cell
is charged a constant number of times.
Note that a cut only introduces a constant number of new cells in the arrangement, so the complexity
stays the same after all cuts induced by the $\psi$-curves.

Any isoline cut on the boundary of $\sigma$ is triggered by a local extremum $p$ at some boundary component.
We charge the cut to the $0$- or $1$-cell that $p$ belongs to.
Note that every vertex is charged for at most $6$ isoline cuts,
because the vertex can only participate in $\binom{4}{2}=6$ bisectors, and thus, $6$ $2$-cells.
Each $1$-cell is only charged a constant number of times, according to Lemma~\ref{lem:local_extrema},
because we cut at most $3$ times for every local extremum of the trisector that contains the $1$-cell.
Since each isoline segment induces only a constant number of additional cells, every cut can be charged
to a Voronoi cell and every cell is charged only a constant number of times, we obtain the result.
\end{proof}

Together with Theorem~\ref{thm:coarrangement}, it follows that we can find a cellular filtration
of size $O(\pi(n))$ whose barcode equals the barcode of the original offset filtration.
The exact value of $\pi(n)$ is far from being settled but we know that it is bounded from below by 
$\Omega(n^2)$ (by a straightforward construction even for n points) and from above by 
$O(n^{3+\eps})$~\cite{Agarwal97computingenvelopes,upper-bound-lower-envelope}, 
so that the obtained filtration is significantly smaller than the unrestricted nerve filtration 
of size $O(n^4)$.

\paragraph{Efficient computation of the stratified refinement}

We describe the procedure for constructing a refined stratification.
We assume that the following primitives are readily available: 
\begin{inparaenum}
\item obtaining all critical values on a trisector;
\item finding intersection points of an isoline and a boundary curve; and
\item a boolean primitive defining a consistent order of any two points on a given isoline.
\end{inparaenum}
As a preprocessing step we cut all boundary curves so that they are $\distance$-monotone.
Since the number of critical values on a trisector is constant, and using the same arguments as 
before, we get that the first primitive can be performed in constant time. The same holds for the 
second primitive, as we make sure that boundary curves are $\distance$-monotone.
Since an isoline is a closed curve, the third primitive can be performed by fixing a reference 
point on the isoline, and then, given 
any two points, decide according to their clockwise order from the reference point. As each isoline is 
composed of a constant number of curves (by our assumption that each polyhedra is of constant size), 
the operation can be performed in constant time.

Fixing a cell $\sigma$, the algorithm uses a sweep-line approach by sweeping an isoline through $\sigma$.
For each critical point $p$, let its \emph{critical value} be $\distance(p)$.
We sort all critical points according to their critical values, and go over the critical points in an 
increasing order. As the sweep progresses, we maintain a list of boundary curves that the 
sweep isoline is currently intersecting. 
The key observation is that when the sweep progresses between two isolines, and since the boundary 
curves do not intersect in their interior, the order of intersection points of the boundary curves with 
the isolines remains the same, with the exception that it might get shifted. However, we can fix the 
reference point in the third primitive so that it remains consistent with respect to the order on the 
previous isolines.
This allows for efficient lookup, insertion, and deletion of curves by maintaining a balanced search 
tree of the curves, sorted by the third primitive.

The algorithm runs in two stages. The first stage creates the unbounded curve $\psi$ by simulating 
combinatorially a curve that has the desired properties, namely, that it is an unbounded curve that 
attains a minimum in $o$ and that is monotone otherwise.
Note that $\psi$ is composed of two $\distance$-monotone unbounded segments. 
We describe how to simulate a single segment using a sweep, and will repeat the procedure for the 
creation of the two segments. 
We start by picking a random intersection point for $\psi$ on the first isoline that we encounter in the 
sweep. Then, at the next isoline, we assume that $\psi$ acts monotonously and pick a random intersection 
point such that $\psi$ does not intersect any curve in the tree. For example, if the current intersection 
point is located between curves $c_1$ and $c_2$ at the current isoline, then we pick the next intersection 
point so that it is still located between $c_1$ and $c_2$ on the next isoline. We update the arrangement 
according to the choice of points by splitting the face in which the curve passes.
A problem may occur when $c_1$ and $c_2$ merge, but as this must happen at a critical point, 
we can simply assume that $\psi$ intersects both $c_1$ and $c_2$ at that critical point, 
and update the arrangement accordingly. At the end of the process we obtain an arrangement that contains 
all the combinatorial information for the cuts induced by $\psi$. Note that as each new edge induced by $\psi$ 
attains a minimum over $\distance$ in its lower endpoint, we also have the critical value of the edge.

The second stage of the algorithm introduces the cuts that are induced by isoline segments.
We again perform a sweep over $\sigma$. 
When the sweep isoline hits a critical point, we look for the two curves 
to the sides of the intersection point in the tree, add the isoline segment bounded by those two 
curves to the arrangement, and split the faces accordingly. We note that it might be that the isoline 
segment goes in the interior of a hole. In order to handle such cases, we keep track for any two 
consecutive curves in the tree whether the area between them belongs to the interior of $\sigma$ or 
not.

Notice that each sweep goes over $O(\pi(n))$ critical values. At each critical value, a constant number 
of lookups, insertions and deletions is performed on the balanced search tree of curves. Also, the size of 
the tree is bounded by the size of the Voronoi diagram. Therefore, for each critical value, we perform 
$O(\log{\pi(n)})$ operations, obtaining the following result.

\begin{theorem}
The offset filtration of convex, disjoint polyhedra in generic position can be computed in 
time $O(\pi(n)\log(\pi(n)))$, excluding the computation time of the Voronoi diagram.
\end{theorem}

Together with the result of~\cite{Agarwal97computingenvelopes}, we obtain an upper bound of 
$O(n^{3+\eps})$ on the expected time.

\section{Convex shapes vs.\ point samples: Experimental comparison}
\label{sec:experiments}

Computing Voronoi diagrams of polyhedra in space is a difficult problem.
While efforts to compute it are underway, so far only restricted cases have been completed; see e.g., \cite{hsh-constructing}.
We therefore restrict our experiments to the planar case.
Still, we show that already in the plane approximating the shapes by point samples introduces 
noise that is hard to distinguish from real small features.

\paragraph{Implementation details} We implemented the restricted nerve algorithm for the two-dimensional case 
as described in Section~\ref{sec:restricted_barcodes_2d}. 
The algorithm requires computing the combinatorial structure of the Voronoi diagram, 
and the critical value of each feature of the diagram, namely what is the point on the feature that attains 
the minimal distance to the sites defining the feature.
We solve the first task by computing the Delaunay graph for all line segments belonging to the input polygons,
using \textsc{Cgal}'s 2D Segment Delaunay Graphs package~\cite{karavelas-cgal}. We restrict our attention
to features of the graph for which no two defining segments belong to 
the same polygon and remove duplicates.
For the critical values, we explicitly compute the actual curves and vertices of the Voronoi diagram
and compute the minimal distance to their nearest polygonal sites. 

For comparison purposes, we also implemented the unrestricted nerve filtration described in 
Section~\ref{sec:barcodes_of_shapes}.
For the computation of critical values, we proceed in a brute-force manner, that is, we
take the minimum distance over all pairs or triplets of line segments of the input polygons.

We also implemented an approximation of the barcode through point samples:
Fixing some $\eps>0$, we calculate a finite point set whose Hausdorff-distance
to the input polygonal shapes is at most $\eps$. 
We achieve that by placing a grid of side length $\sqrt{2}\eps$ in the plane
and taking as our sample the centers of all grid cells which are intersected by some input polygon.
We can compute this sample efficiently by taking the Minkowski sum of each polygon and a square with side length $\sqrt{2} \eps$, 
and performing batch point location queries for a grid of points in the bounding rectangle of that polygon.
We compute the alpha-filtration on the point samples, using \textsc{Cgal}'s 2D Triangulation package \cite{yvinec-2d}.
The size of the resulting filtration is linear in the number of sampled points.

In all variants, after computing the filtration, we obtain the barcode using the PHAT 
library~\cite{bkrw-phat}.

\paragraph{Time analysis}
\begin{table*}[t,b]
\footnotesize
\begin{center}
	\begin{tabular}{|c c||c c c c c||}
		\hline
		& &\multicolumn{5}{c||}{Number of vertices}\\
		Approach 			
		&	 &	$42$ 	&	$213$	& $1060$	& $2104$ & $4217$\\
		\hline
		{\multirow{4}{*}{Restricted nerve}}
		&	Filtration time		&	0.042 	&	0.237	&	1.285		&	2.716	&	5.693	\\
		&	Persistence time	&	0	 	&	0		&	0	 	&	0	&	0 	\\
		&	Total time			&	\textbf{0.042}	&	\textbf{0.237}	&	\textbf{1.285}	&	\textbf{2.716}	&	\textbf{5.693} 	\\
		&	Filtration size		&	45		&	267	&	1473		&	2963	&	5949	\\
		\hline
		{\multirow{4}{*}{Unrestricted nerve}}
		&	Filtration time		&	0.602 	&	69.13	&	9920		&	-		&	-		\\
		&	Persistence time	&	0	 	&	0.033	&	23.11	 	&	-		&	-	 	\\
		&	Total time			&	\textbf{0.602}	&	\textbf{69.16}	&	\textbf{9943}	&	\textbf{-}		&	\textbf{-} 		\\
		&	Filtration size		&	175		&	20875	&	2604375		&	-		&	-		\\
		\hline
		{\multirow{4}{*}{Point sample ($\eps = 1$)}}
		&	Filtration time		&	0.013 	&	0.047	&	0.217		&	0.435	&	0.805	\\
		&	Persistence time	&	0	 	&	0		&	0.001	 	&	0.001	&	0.003	\\
		&	Total time			&	\textbf{0.013}	&	\textbf{0.047}	&	\textbf{0.217}	&	\textbf{0.437}	&	\textbf{0.808} 	\\
		&	Filtration size		&	803		&	2101	&	9021		&	14235	&	22393	\\
		\hline
		{\multirow{4}{*}{Point sample ($\eps = 0.5$)}}
		&	Filtration time		&	0.023 	&	0.07	&	0.312		&	0.594	&	1.045	\\
		&	Persistence time	&	0 		&	0		&	0.003		&	0.006	&	0.01 	\\
		&	Total time			&	\textbf{0.024}	&	\textbf{0.07}	&	\textbf{0.315}	&	\textbf{0.6}	&	\textbf{1.055} 	\\
		&	Filtration size		&	2577	&	6707	&	28273		&	45635	&	72275	\\
		\hline
		{\multirow{4}{*}{Point sample ($\eps = 0.1$)}}
		&	Filtration time		&	0.304 	&	0.765	&	3.22		&	5.106	&	7.975	\\
		&	Persistence time	&	0.007 	&	0.016	&	0.081	 	&	0.132	&	0.213 	\\
		&	Total time			&	\textbf{0.310}	&	\textbf{0.781}	&	\textbf{3.3}	&	\textbf{5.238}	&	\textbf{8.188} 	\\
		&	Filtration size		&	48741	&	120359	&	505833		&	792395	&	1226827	\\
		\hline
	\end{tabular}		
\end{center}
\caption{\sf Running time and filtration size with respect to input size. The results are averaged over 5 runs.
Some results for the unrestricted nerve approach are omitted as execution did not finish within a reasonable time. 
Times are measured in seconds, and are highlighted in bold font.}
\label{tbl:time_comprisons}
\end{table*}

We compare the three approaches for several inputs of increasing sizes, and report on their running times
and the size of the filtration. 
For generating the input, we considered a square of side length~100. We added random polygons inside the square using 
the following repetitive process. We randomly pick a point inside the square, and create a rectangle centered at that 
point with a random width and height. We then randomly select 5 points inside that rectangle and take the polygon to be 
their convex hull. If the polygon does not intersect any of the polygons that were previously added, we add it to our input.
An illustration of an input is shown in Figure \ref{fig:quality_of_barcode}.
All experiments were run on a 3.4GHz Intel Core i5 processor with 8GB of memory. 

The experimental results are shown in Table~\ref{tbl:time_comprisons}.
We observe that computing the filtration takes much more time than converting
the filtration into a barcode. This is in sync with the common observation that despite
the worst-case cubical complexity, the barcode computation scales well in practice.
Moreover, we get the expected result that the unrestricted nerve yields large filtrations
and high running times compared to the restricted nerve approach. In particular, 
for sufficiently large inputs, the time to compute the barcode
from the unrestricted filtration is an order of magnitude higher than the total time in the restricted case,
so even a smarter way for computing critical values will not be helpful.

Comparing the running time of the restricted nerve approach with that of the point sampling approach is difficult as
it depends on the choice of~$\eps$. What approximation quality is reasonable generally depends on the application and 
the input at hand. However, we can observe that the overhead of computing the Voronoi diagrams of convex polygons
instead of points is rather small, and we get the exact barcode in the same time as it would take to compute 
a ``reasonable'' approximation of the input. 
Of course, it should be admitted that our simple point-sampling approach could be significantly improved
by sampling only the boundary of polygons instead, and even more by making the sample density adaptive to the 
\emph{local feature size}~\cite{ruppert-delaunay} of the polygons. This however would require post-processing of the barcode to filter out topological features in the interior of polygons.

\paragraph{Quality of the barcodes}

\begin{figure}[t,b]
\centering
	\subfloat[][Polygons used as input]{
		\includegraphics[width=0.25\textwidth] {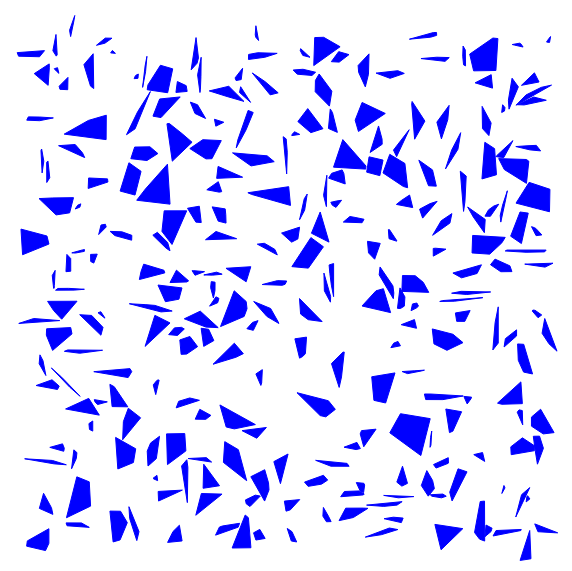}
		\label{subfig:exact_input}
	}
	\hspace{1cm}
	\subfloat[][Exact barcode of the polygons]{
		\includegraphics[width=0.25\textwidth] {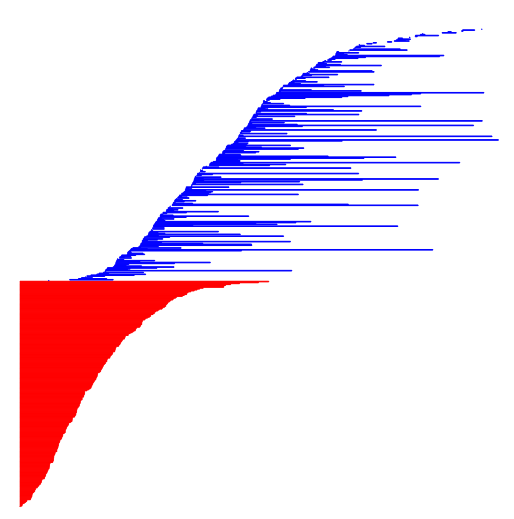}
		\label{subfig:exact_barcode}
	}
	
	\subfloat[][A point set approximation of the input polygons]{
		\includegraphics[width=0.25\textwidth]{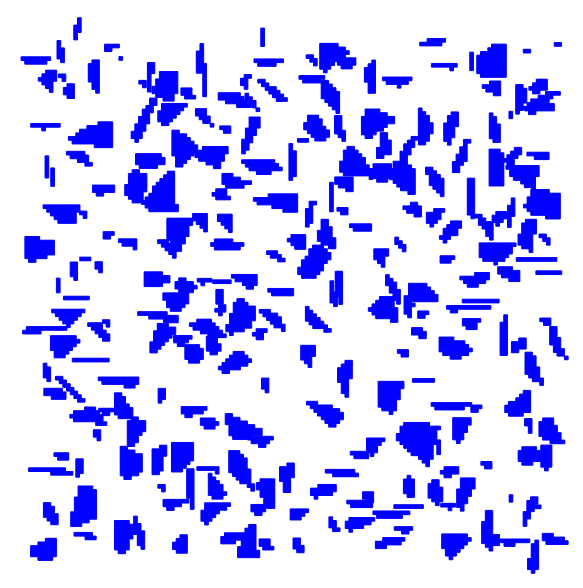}
		\label{subfig:approx_input}
	}
	\hspace{1cm}
	\subfloat[][Barcode of the point set]{
		\includegraphics[width=0.25\textwidth]{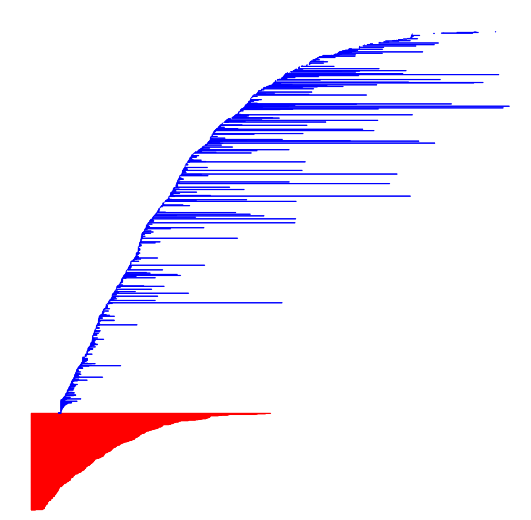}
		\label{subfig:approx_barcode}
	}
	\caption{
	A comparison between an exact barcode and an approximated barcode. 
	An illustration of polygons used as input for our experiments is shown in~\protect\subref{subfig:exact_input}. 
	There are~250 polygons with a total of~1060 vertices, inside a square of side length~100. The exact barcode of the
	polygons, computed using the restricted nerve method, is shown in~\protect\subref{subfig:exact_barcode}. Red bars represent 
	connected components and blue bars represent holes. It can be observed that all connected components are born
	at offset~$0$, since no connected components are created as the offset increases. When increasing the offset,
	connected components merge and therefore die, and holes are created and then die as they get filled.
	\protect\subref{subfig:approx_input} shows an illustration of a point set approximation of the input polygons, with
	$\eps = 0.5$. Each point is displayed as a pixel of size $\sqrt{2}\eps$. In~\protect\subref{subfig:approx_barcode}
	we see the barcode of the approximation point set, which is an approximation of the exact barcode of the polygons. The approximated barcode contains many more bars compared to the exact barcode. Noise, in the form of short bars, can be 
	observed throughout the entire barcode.
	}
	\label{fig:quality_of_barcode}
	\vspace{-3mm}
\end{figure}

The barcode obtained using the point sampling approach approximates the actual barcode of the input. 
The stability theorem~\cite{ceh-stability} ensures that the two barcodes have a distance of at most~$\eps$.
Still, the question arises how much noise is introduced by the approximation. 
We remind the reader about our motivation to produce barcodes of shape: We want to identify offset values
that minimize the number of short-lived homological features.
Figure~\ref{fig:quality_of_barcode} shows an exact barcode and an approximated barcode for the
same input. We see that the exact barcode contains some short bars throughout the entire barcode. The
approximated barcode contains many more such bars, some originating from real features in the input, and some  
are artifacts of the approximation. 
It seems (at best) very difficult to identify the real features from the approximate barcode,
which speaks in favor of using the exact approach that we have taken in this paper 
in such types of applications.

\section{Conclusion}
\label{sec:conclusion}

In this paper, we have defined cellular filtrations that yield
the same barcode as the offset filtrations of convex polygons
in $\R^2$, and convex polyhedra in $\R^3$. Our filtrations generalize
the alpha filtrations for point sets and have a size proportional
to the size of the Voronoi diagram of the involved sites.
While the approach in $\R^2$ simply consists of taking the nerve 
of the Voronoi regions, we had to cut lower-dimensional cells
in $\R^3$ to obtain such a filtration.

We restrict to the case of disjoint polyhedra for simplicity;
in the light of our intended application to 3D printing, an extension
to interior-disjoint polyhedra is desirable.
Assuming disjointness, however, is not a serious restriction,
as we can conceptually think of each polyhedron 
to be shrunk by an infinitesimal value, yielding a barcode that
is arbitrarily close to the exact one. The details of this approach
are left for an extended version of this paper.

We suspect that our approach
from Section~\ref{sec:higher_dim} generalizes to $\R^d$
for $d\geq 3$
and yields a filtration of same asymptotic complexity
as the Voronoi diagram (for a constant dimension $d$).
The proof of this statement, however, requires several non-trivial
extensions of our results, for instance, a generalization of pseudodisk property (Theorem~\ref{thm:pseudodisk_maintext})
in arbitrary dimensions.

We are in the process of implementing our refinement approach for the case of lines in $\R^3$,
using the implementation of~\cite{hsh-constructing}.
An interesting question in this context is whether for this special case, the Voronoi diagram
is already stratified. It is not even clear to us whether in this case,
an analogue of Theorem~\ref{thm:1-barcode_equivalence} holds, that is, whether the the nerve
of the Voronoi regions yields an equivalent filtration.

\paragraph{Acknowledgment} 
The authors thank Micha Sharir for suggesting the outline of a proof of the pseudo-sphere property in three-dimensional space.
The authors thank Michael Sagraloff for helpful discussions during this work.

\newpage

\begin{appendix}

\section{Proof of Theorem~\ref{thm:1-barcode_equivalence}: Tisk-theory}
\label{app:monster_proof}

The main goal of this section is to give a proof of Theorem~\ref{thm:1-barcode_equivalence}. 

For that, we first slightly abstract from offsets and Voronoi regions, and define a class of objects that includes all of them.

\begin{definition}
A \emph{piecewise-algebraic path} is a simple path in $\R^2$ which consists
of a finite number of arcs, where each arc is a semi-algebraic curve.
A \emph{piecewise-algebraic loop} is such a path that is homeomorphic to a
circle.
A set $D\subset\R^2$ is a \emph{bounded tisk} if $D$ is the closed bounded region 
induced by a piecewise-algebraic loop. $D$ is an \emph{unbounded tisk} if
it is the closed region induced by an unbounded piecewise-algebraic path.
Two tisks $A$, $B$ with boundary curves $a$, $b$, are called \emph{interior-disjoint} if $A\cap B=a\cap b$.
\end{definition}

We remark that every Voronoi region as well as any restricted offset is a tisk,
because their boundaries consist of line segments, circular arcs, and parabolic arcs.

We consider the intersection of two interior-disjoint tisks $A$ and $B$. Let $a$ and $b$ denote the
boundary curves of $A$ and of $B$ respectively. The curve $a$ is partitioned into interior-disjoint
segments, alternating between segments that belong to $a\cap b$ (such a segment may degenerate to a point), 
and segments that belong
to $a\setminus b$. The symmetric property holds for $b$. 
The partition points are denoted by $V_{ab}$.
Note that $V_{ab}$ is empty if and only if $A$ and $B$ are non-intersecting, or both are unbounded
and their union is the whole space. Note also that by choosing an arbitrary orientation on $a$
or on $b$, we can define a total order on $V_{ab}$. These definitions are illustrated
in Figure~\ref{fig:interior_disjoint_tisks}.

\begin{figure}[htb]
\centering
\includegraphics[height=4cm]{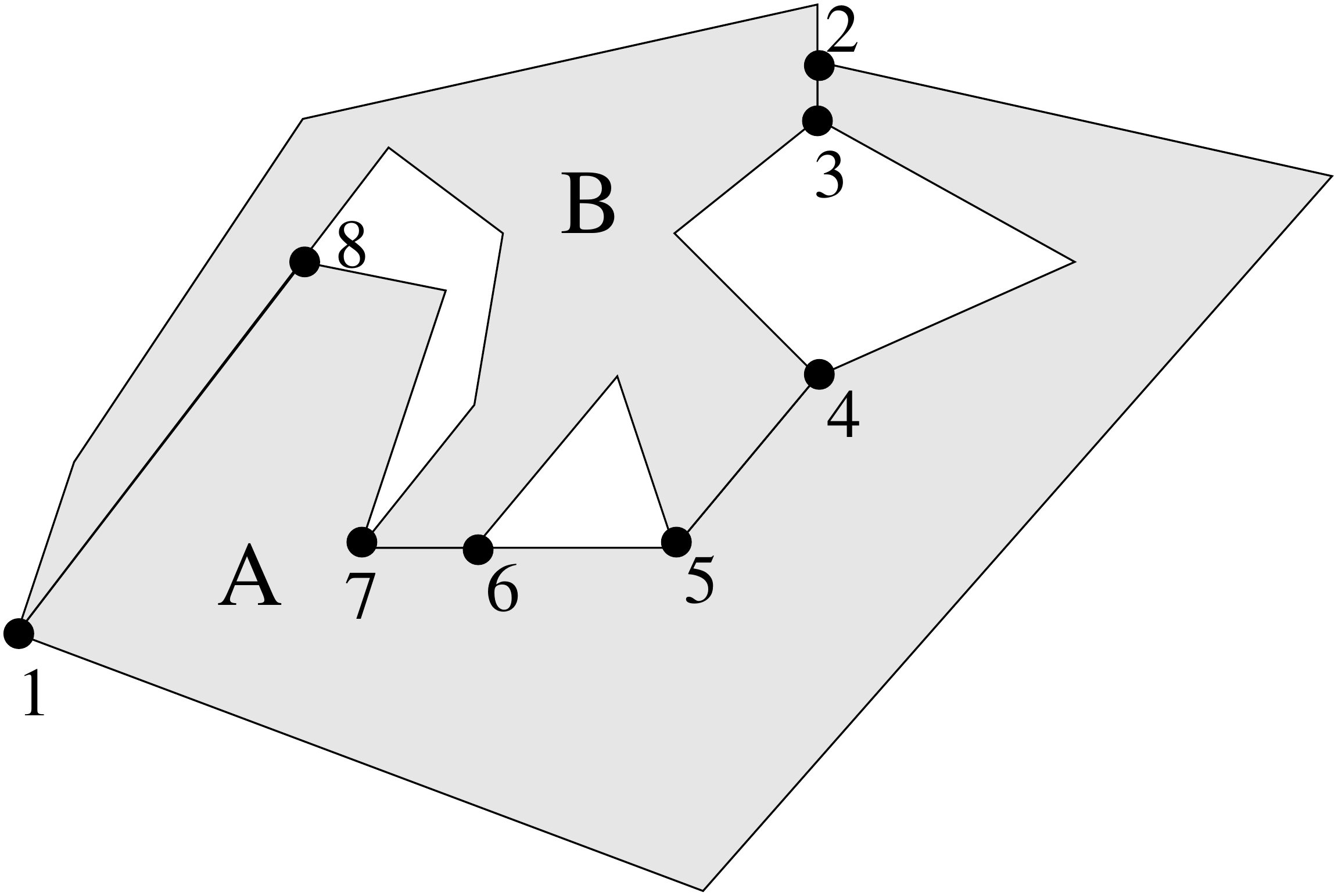}
\hspace{1cm}
\includegraphics[height=4cm]{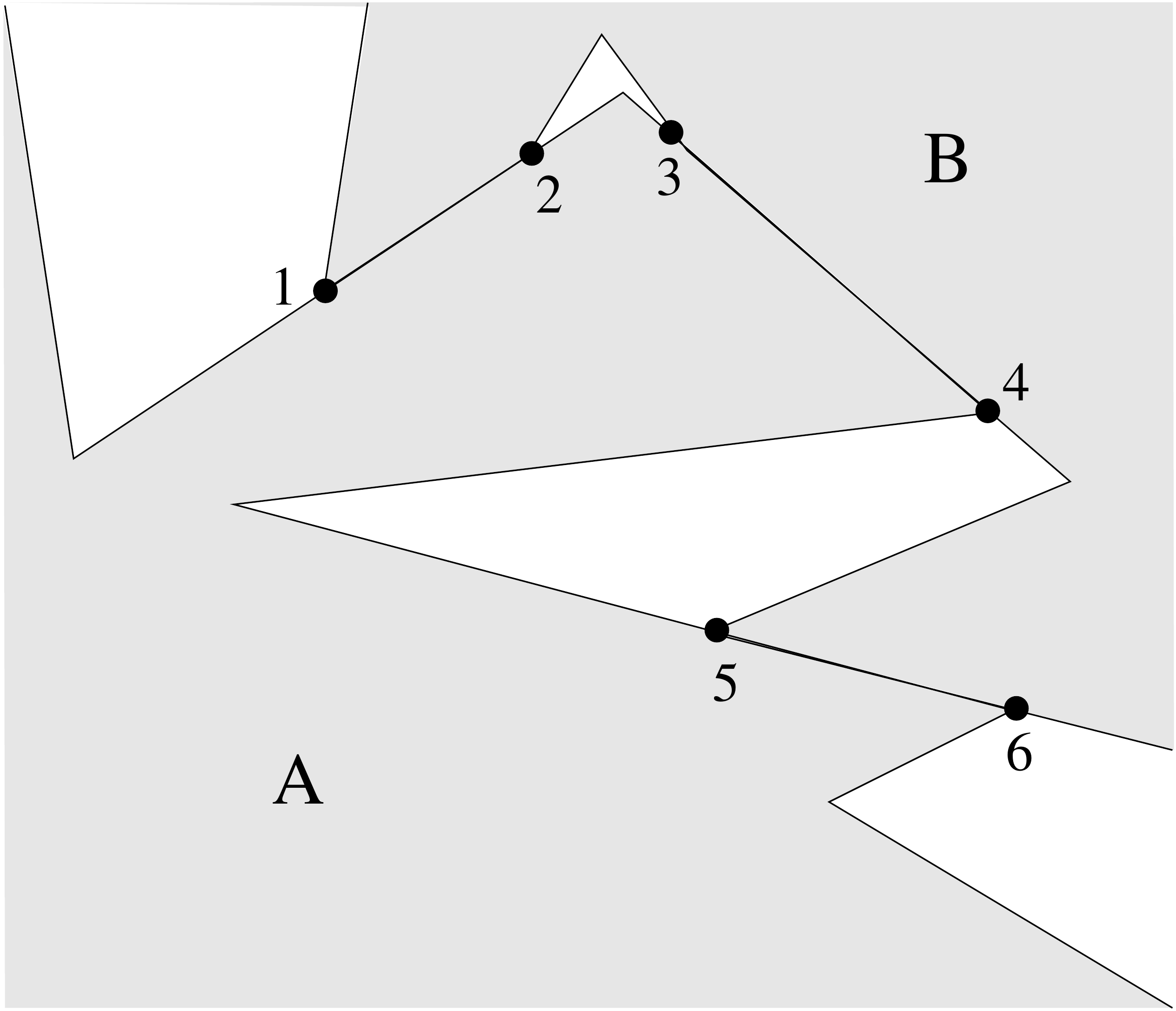}
\caption{Examples of interior-disjoint intersecting tisks. 
On the left, their complement is composed into three surrounded regions and one unbounded region.
On the right, the complement is composed into two surrounded regions and two unbounded regions.
The points on $V_{ab}$ are highlighted, and a possible ordering is given.}
\label{fig:interior_disjoint_tisks}
\end{figure}

\begin{definition}
The connected components of $\Clos(\R^2\setminus{A\cup B})$ are called \emph{complementary regions}
of $(A,B)$. A \emph{surrounded region} is a complementary region that is bounded.
\end{definition}

Note that this definition of surrounded regions agrees with the definition given in Section~\ref{sec:restricted_barcodes}
for the special case that tisks are Voronoi regions.
Every surrounded region is a tisk whose boundary loop involves two consecutive points of $V_{ab}$
connected by one segment in $a\setminus b$ and one segment in $b\setminus a$.
An unbounded complementary region might or might not be a tisk, depending on the whether 
$A$ and $B$ are bounded. For instance, if both $A$ and $B$ are bounded, 
there is only one unbounded complementary region, and it is the complement of a tisk.
We can easily see that every surrounded region of $(A,B)$ contains exactly two points of $V_{ab}$.
Every unbounded complementary region contains one or two points of $V_{ab}$; more precisely,
if there is only one unbounded region, it contains two points of $V_{ab}$, and if there
are two such regions, they contain one point of $V_{ab}$ each.
See again Figure~\ref{fig:interior_disjoint_tisks} for illustrations of these concepts.

\begin{lemma}
Assume that $a\cap b=A\cap B$ is non-empty.
Then, it is contractible if and only if
$(A,B)$ does not induce a surrounded region.
\end{lemma}

\begin{proof}
``$\Leftarrow$'' follows directly from the definition: Assume that $a\cap b$ is non-contractible.
Then, $a\cap b$ decomposes into at least two connected components, and there exist 
consecutive points $v,v'\in V_{ab}$ which are connected by segments in $a\setminus b$ and $b\setminus a$.
The tisk enclosed by this loop is a surrounded region.

For ``$\Rightarrow$'', contractibility implies that $V_{ab}$ consists of exactly two points, 
and both lie in the unbounded complementary region, so there are no
points in $V_{ab}$ left to form a surrounded region.
\end{proof}

Next, we consider a collection $\mathcal{S}$ of pairwise interior-disjoint tisks. 
We assume generic position, that is, no more than three tisks intersect
in a common point.
Three tisks $A,B,C$ of $\mathcal{S}$ 
with boundary curves $a,b,c$ intersect in at most two points; this follows from the observation
that $a\cap b\cap c\subseteq V_{ab}$ (because any intersection
of $c$ in the interior of a segment of $a\cap b$ would imply that $C$
is not interior disjoint to $A$ or $B$), combined with the fact that $C$
must be contained in some complementary region of $(A,B)$, each of which
contains at most $2$ points of $V_{ab}$.
\begin{figure}[htb]
\centering
\includegraphics[width=4cm]{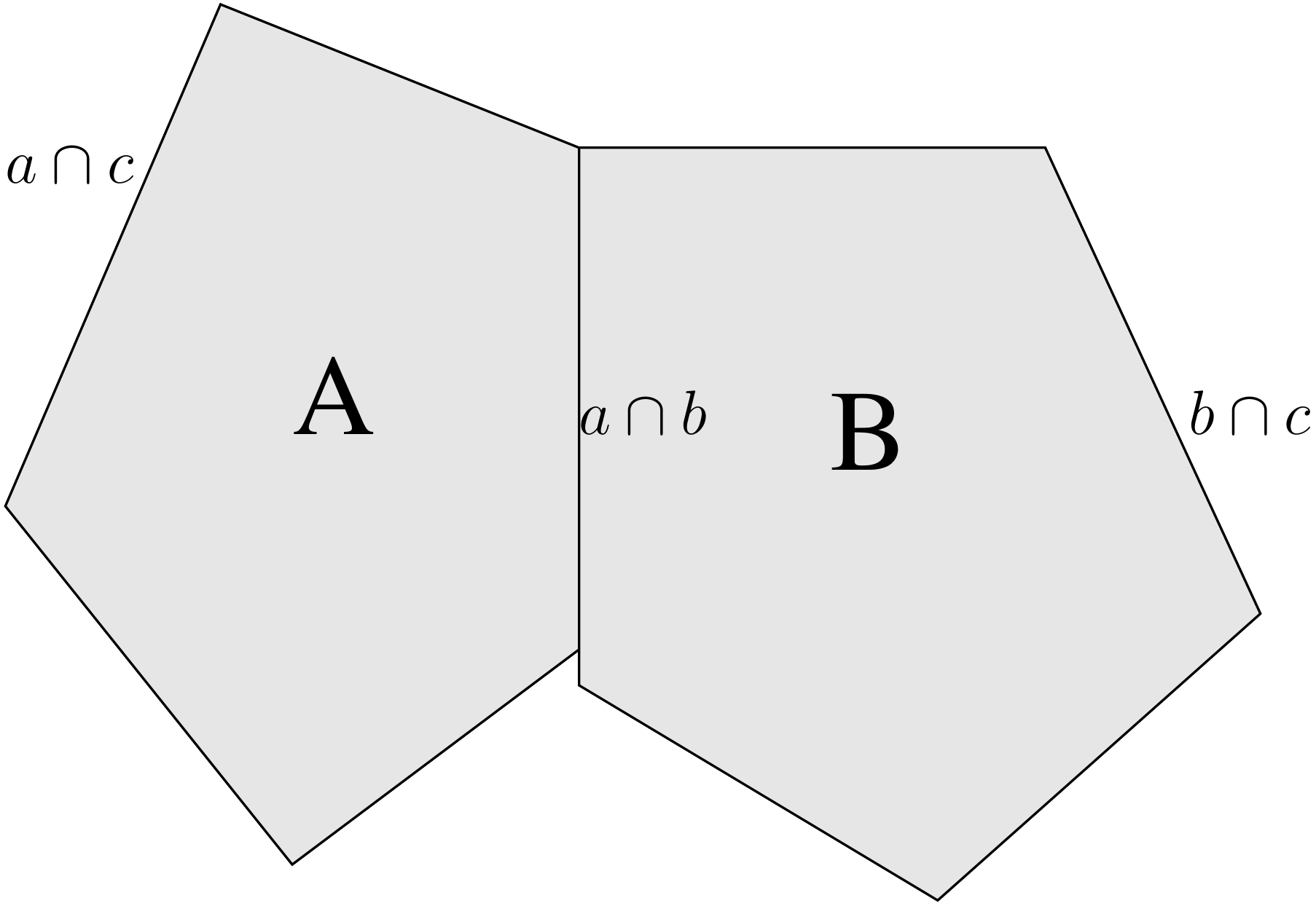}
\caption{Illustration of the proof of Lemma~\ref{lem:edges_suffice_for_nerve_theorem}. 
The set $C$ has to be bounded by the outer cycle, which implies that either $C$ has an inner hole
(contradicting the tisk-property) or it contains $A$ and $B$ (contradicting interior-disjointness).}
\label{fig:one_intersection_point}
\end{figure}

\begin{lemma}\label{lem:edges_suffice_for_nerve_theorem}
If any pair in $\mathcal{S}$ is empty or contractible, then every triple in $\mathcal{S}$
is empty or a single point. 
\end{lemma}

\begin{proof}
Consider a triple $A,B,C$ with boundary curves $a,b,c$ and non-empty intersection. 
We assume that $a\cap b$
is contractible. That implies that $|V_{ab}|\leq 2$. Wlog, we can assume that $|V_{ab}|=2$,
because otherwise, $a\cap b\cap c\subseteq V_{ab}$ contains of at most one point
and we are done. Let $V_{ab}=\{v_1,v_2\}$. There exists a segment of $a\cap b$
connecting $v_1$ to $v_2$. Assume first that $A$ is unbounded: then, there exist two
segments in $a\setminus b$, one emenating at $v_1$, one at $v_2$. Since $c$ cannot intersect
$a$ in the interior of $a\cap b$, $a\cap c$ must be a subset of $a\setminus b$.
However, $a\cap c$ contains $v_1$ and $v_2$, therefore, it consist of at least two connected
components. This is a contradiction to the assumption that $a\cap c$ is contractible.
A symmetric argument implies that $B$ must be bounded.

We are left with the case that both $A$ and $B$ are bounded. Then $a\setminus b$
and $b\setminus a$ are segments that connect $v_1$ and $v_2$. Since $a\cap c$
and $b\cap c$ are contractible, it follows that the loop of $c$ is the union of
these segments. However, that loop encloses both $A$ and $B$, which contradicts
the interior disjointness the tisks.
\end{proof}

\begin{figure}[htb]
\centering
\includegraphics[width=4cm]{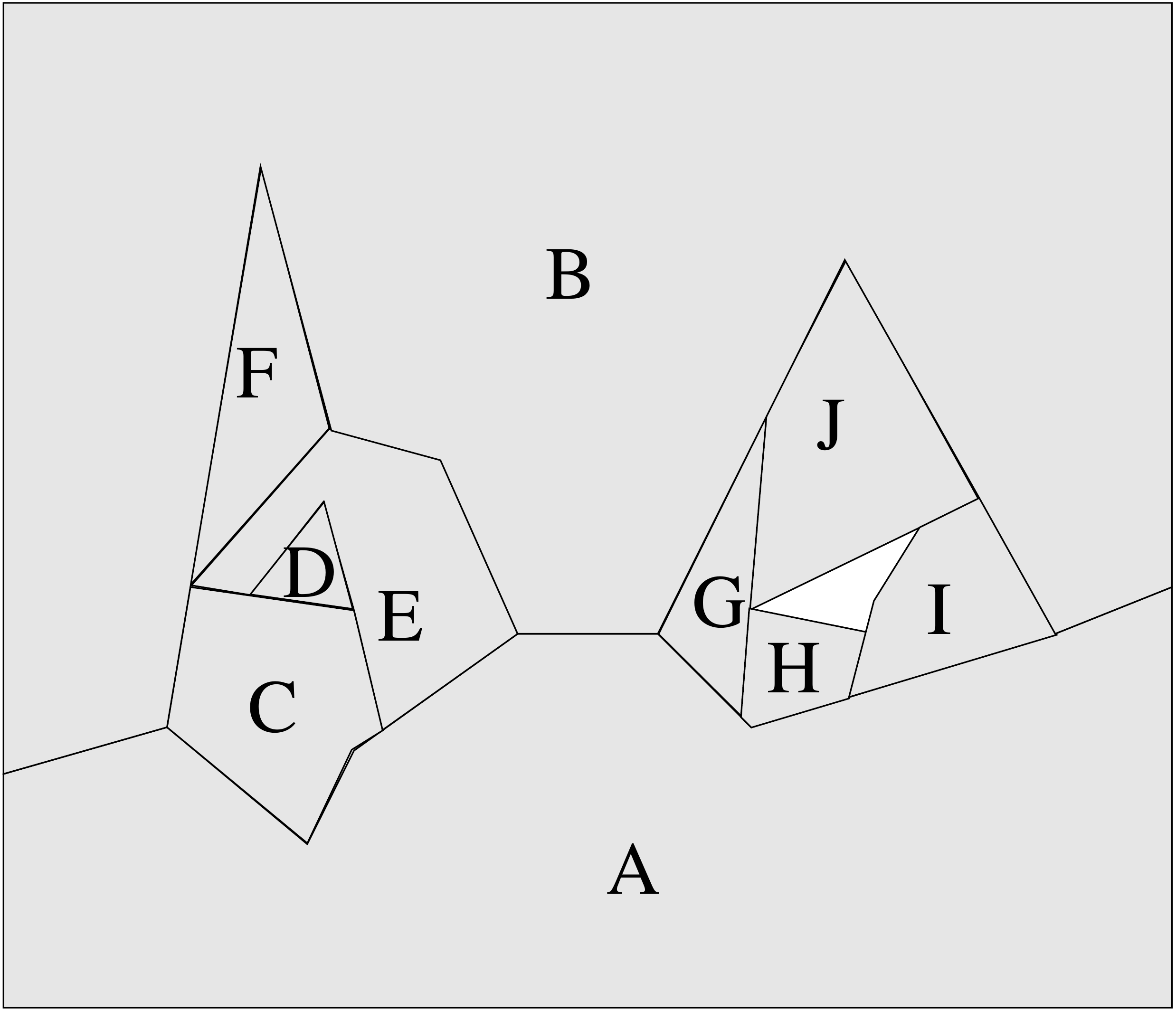}
\caption{$A$ and $B$ induce two surrounded regions. The left one has the four members $C,D,E,F$,
the right one the four members $G,H,I,J$. The left region is filled but not simple, because $C$
and $E$ induce a nested surrounded region. This nested region, in turn, is filled and simple.
The right region is simple, but not filled; note that the uncovered hole is not a surrounded region.}
\label{fig:members}
\end{figure}

For the upcoming definitions, see also Figure~\ref{fig:members}.
For a surrounded region $R$ induced by $A$ and $B$,
we call a tisk $S$ in $\mathcal{S}$ a \emph{member} of $R$
if $S\subseteq R$. Note that a surrounded region might have an arbitrary number of 
members, including no member at all.
We let $M_R$ denote the set of members of $R$, and 
$M_R^{\mathrm{ext}}:=M_R\cup\{A,B\}$ the set of \emph{extended members}.
We call a surrounded region $R$ \emph{filled}, if the union of its members equals $R$
(in this case, $R$ must have at least one member).
We define a \emph{simple surrounded region} of $\mathcal{S}$ to be a surrounded region
such that no other surrounded region is contained in it. 

The following property
will be of special importance; we will refer to it as the ``guarding principle''.

\begin{lemma}
For a collection of tisks $\mathcal{S}$ in generic position and a surrounded region $R$,
there is no intersection of a member of $R$ with an element of $\mathcal{S}\setminus M_R^{\mathrm{ext}}$.
\end{lemma}
\begin{proof}
Let $R$ be induced by the tisks $A$ and $B$, and let $a$, $b$ their boundaries.
Assume by contradiction the existence of such an intersection. 
Then, there must be a tisk $C\in M_R^{\mathrm{ext}}$ that intersects the boundary of $R$.
It is easy to see that this intersection cannot take place within $a\setminus b$
(because $A$ and $C$ are interior-disjoint and $a$ is simple), and the same way
for $b\setminus a$. It follows that $C$ intersects the boundary of $R$ at $v\in a\cap b$. 
However, by assumption, $C$ intersects a fourth tisk $D\in M_R$ at point $v$.
So, $v$ is an intersection of four tisks which contradicts our genericity assumption.
\end{proof}

An equivalent statement is that if $R$ is induced by $A$ and $B$, the set $\{A,B\}$ constitutes
a separator in $\nerve\mathcal{S}$, separating the elements within $R$ and the elements
outside of $R$.

We now state the first main theorem that will be needed for the proof of Theorem~\ref{thm:1-barcode_equivalence}.
As usual, we let $|\mathcal{S}|$ denote the underlying space of $\mathcal{S}$, which is the union over all tisks in $\mathcal{S}$.

\begin{theorem}
Let $\mathcal{S}$ be a collection of tisks in generic position such that every surrounded region is filled.
Then $H_1(\nerve(\mathcal{S}))=H_1(|\mathcal{S}|)$.
\end{theorem}

The idea of the proof is to ``clear out'' surrounded regions one after the other by removing the members 
within a surrounded region and charging the surrounded region to one of the two tisks that surround it.
These operations do not change the underlying space. Moreover, as we will show, there cannot be any
non-trivial $1$-homology in the nerve of a surrounded region. Finally, by the guarding principle,
removing the member of a surrounded region does not affect the connectivity of remaining tisks.
These properties will be enough to ensure the isomorphism. We give the details of the single steps next.

\begin{figure}[htb]
\centering
\includegraphics[width=8cm]{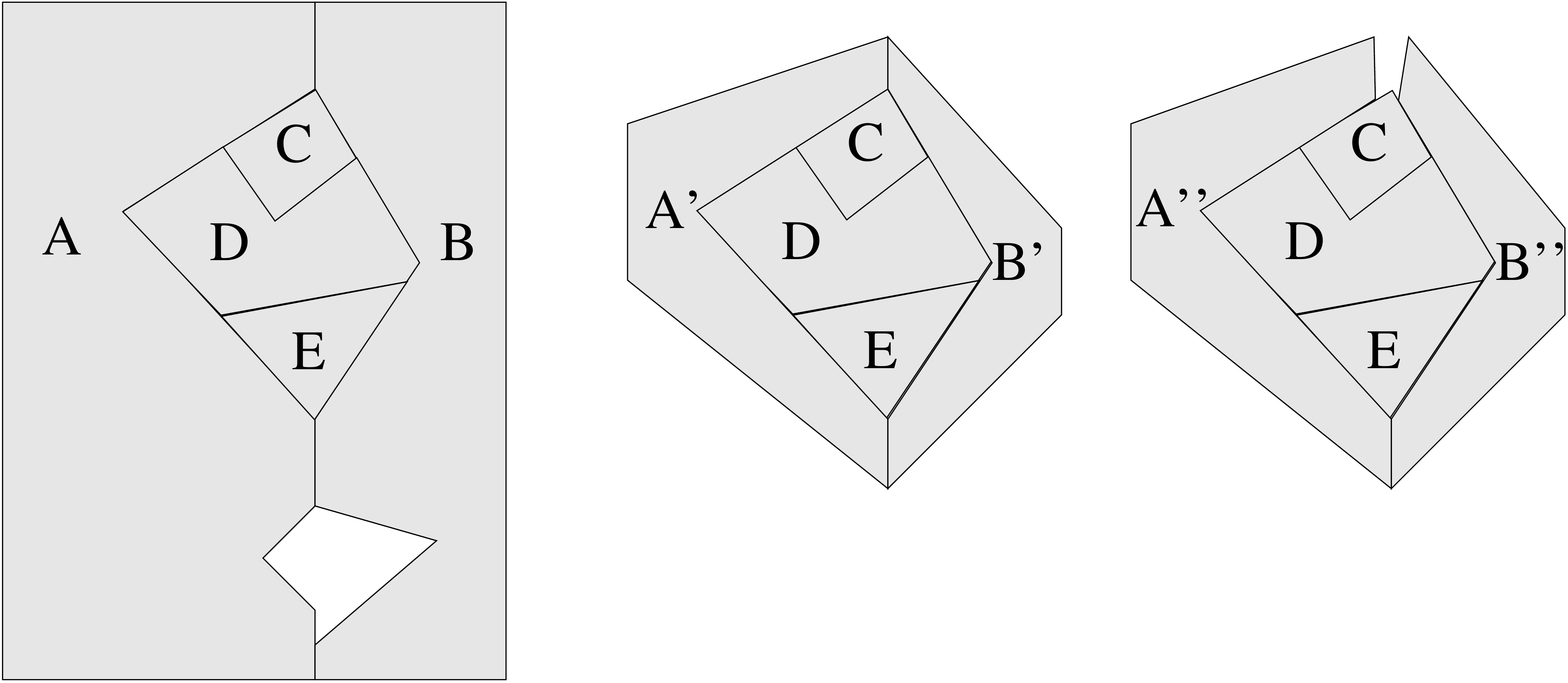}
\caption{Illustration of the transformation in the proof of Lemma~\ref{lem:extended_members_homology}.
On the left, we see a simple surrounded and filled region $R$ with members $C$, $D$, $E$.
In the middle, we shrink $A$ and $B$ to local neighborhoods around $R$, ignoring possible further
surrounded regions. On the right, we disconnect one of the two intersections of $A'$ and $B'$
without disconnecting $(A',C)$ or $(B',C)$. The obtained subdivision has only contractible intersections.}
\label{fig:surrounded_surgery}
\end{figure}

\begin{lemma}\label{lem:extended_members_homology}
Let $R$ be a surrounded region that is simple and filled. Then, $H_1(\nerve M_R^{\mathrm{ext}})=0$.
\end{lemma}
\begin{proof}
Let the surrounded region $R$ be induced by $(A,B)$, 
and we set $M:=M_R^{\mathrm{ext}}$ for notational convenience.
First of all, note that we can restrict $A$ and $B$ to a local neighborhood around $R$
without changing the nerve of $M$. Formally, replace $A$ by $A'$, which is the intersection
of $A$ with an $\eps$-offset of $R$ (for $\eps>0$ small enough), and same for $B$.
We set $M':=M_R\cup\{A',B'\}$. Clearly, $\nerve M=\nerve M'$.
See Figure~\ref{fig:surrounded_surgery} (middle).

We want to prove the claim using the Nerve theorem; however, the intersection of $A'$ and $B'$
is non-contractible. We perform another local surgery to avoid this problem:
Let $v$ be one of the two intersection points of $A\cap B\cap R$, and let $C\in M_R$
be the member adjacent to this point ($C$ exists because $R$ is filled, and $C$ is unique
because only three tisks intersect in one point).
We can separate $A'$ and $B'$ locally around $v$ with a small distance while leaving the
pairwise intersections with $C$ intact (again, because we assume non-degeneracy, $A'\cap C$ is a
non-degenerate segment). Let $A''$, $B''$ be the replacements, and let $M'':=M_R\cup\{A'',B''\}$; 
note that $A''$ and $B''$ are still connected at the second intersection point of $A\cap B\cap R$;
it follows that $\nerve M''$ has the same edges of $\nerve M$; in fact, the nerves are the same,
except that the triangle $ABC\in\nerve M$ might or might not have a counterpart in $\nerve M''$.
See Figure~\ref{fig:surrounded_surgery} (right).

It is enough to show that $H_1(\nerve M'')=0$.
Any pair of tisks in $M''$ has a non-contractible intersection:
we explicitly constructed $A''$ and $B''$ to have non-contractible intersection, 
and if any pair in $M''$
had a non-contractible intersection, it would introduce a surrounded region inside $R$
which contradicts the assumption that $R$ is simple. 
By Lemma~\ref{lem:edges_suffice_for_nerve_theorem}, 
this implies $M''$ is a good partition 
and the Nerve theorem applies. 
So, $\nerve M''$ is homotopically equivalent to $|M''|$,
which is a topological disk because $R$ is filled.
It follows that $H_1(\nerve M'')=0$, as required.
\end{proof}

Recall that $M_R$ is the set of member of a surrounded region $R$
and let $\mathcal{S}_R$ denote the set $\mathcal{S}\setminus M_R$. 
We define a map
\[\phi_R:\mathcal{S}\rightarrow\mathcal{S}_R\]
mapping each member of $R$ to $A$, and any other tisk to itself.
We have the induced simplicial map
\[\nerve(\mathcal{S})\rightarrow\nerve(\mathcal{S}_R),\]
which maps a simplex $\sigma=(S_0,\ldots,S_k)\in\nerve(\mathcal{S})$
with $S_i\in\mathcal{S}$ to the simplex $(\phi_R(S_0),\ldots,\phi_R(S_k))$.
Slightly abusing notation, we also write $\phi_R$ for this map on the nerve level.
Note that $\phi_R(S_i)$ might be equal to $\phi_R(S_j)$, so $\phi_R$ may map a $k$-simplex to an $\ell$-simplex with $\ell\leq k$.

We need to argue that $\phi_R$ is well-defined, that is, $\phi_R(\sigma)\in\nerve(\mathcal{S}_R)$. For that, observe
that if $\sigma$ does not contain any member of $R$, it stays in the nerve when removing $R$, and $\phi_R(\sigma)=\sigma$.
On the other hand, if $\sigma$ contains any member of $R$, it can only contain extended members of $R$ by the guarding principle. 
Therefore, $\phi_R(\sigma)=(A)$ or $\phi_R(\sigma)=(A,B)$, and both are in $\nerve(\mathcal{S}_R)$.

Being a simplicial map, $\phi_R$ induces a map
\[\phi_R^\ast: H_1(\nerve(\mathcal{S}))\rightarrow H_1(\nerve(\mathcal{S}_R))\]
of homology groups. We show that this map is an isomorphism if
$R$ is simple and filled.

\begin{figure}[htb]
\centering
\includegraphics[width=6cm]{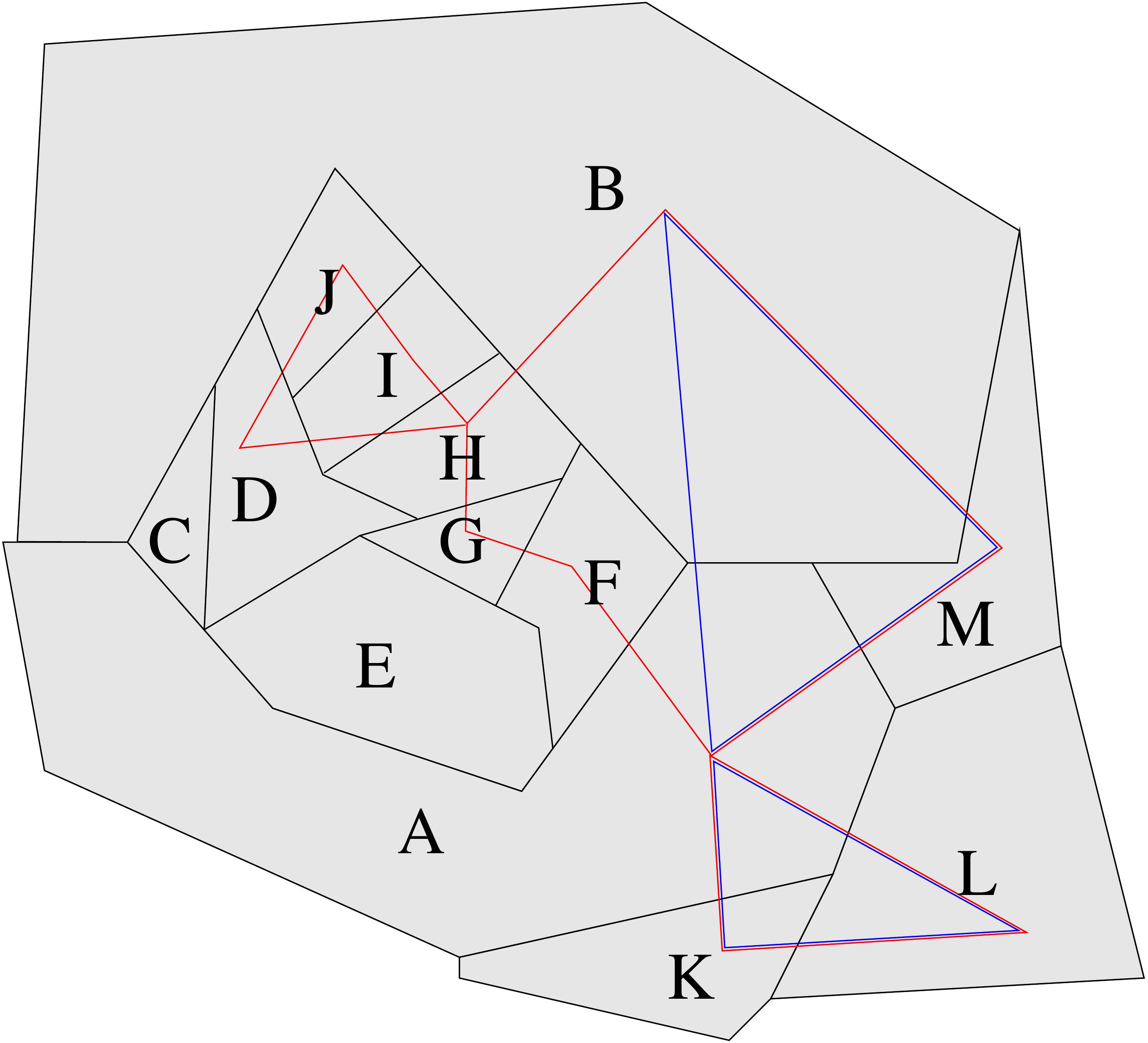}
\caption{Illustration of the transformation in the proof of Lemma~\ref{lem:simple_iso_lemma}.
$A$ and $B$ induce a simple and filled surrounded region $R$ with members $C,D,\ldots,J$.
The red curve represents a cycle $c$ in the $\nerve(\mathcal{S})$: the edges in the cycle
are defined by the intersections that the curve crosses. In this example, the cycle
consists of three loops, each of a different type: The loop $DJ+JI+IH+HD$ consists
of extended members of $R$ (type (1)) is transformed into the trivial cycle. The loop
$AK+KL+LA$ consists of elements of $\mathcal{S}\setminus M_R$ and is thus preserved. The
third loop $AF+FG+GH+HB+BM+MA$ is of mixed type. The path $AF+FG+GH+HB$ consists only
of extended members and is replaced by $AB$. The blue curve shows the resulting cycle $c'$.}
\label{fig:cycle_transform}
\end{figure}

\begin{lemma}\label{lem:simple_iso_lemma}
If a surrounded region $R$ is simple and filled, the map $\phi_R^\ast$
is an isomorphism.
\end{lemma}
\begin{proof}
$\mathcal{S}_R\subseteq\mathcal{S}$ implies immediately that the map is surjective.
The \emph{support} of a $d$-chain is the union of all vertices that are boundary vertices
of at least one simplex in the chain.
For injectivity, we claim for any $1$-cycle $c$ in $\nerve\mathcal{S}$, there is an homologous cycle
that is only supported by vertices in $\mathcal{S}_R$. Indeed, this statement implies injectivity:
Assume that $[c_1], [c_2]\in H_1(\nerve(\mathcal{S}))$ are such that $\phi_R^\ast([c_1])=\phi_R^\ast([c_2])$.
Then, we can replace the representatives $c_1$, $c_2$ by $c_1', c_2'$ supported  by $\mathcal{S}_R$.
Since $\phi_R$ is the identity on $\mathcal{S}_R$, we have
\[[c_1]=[c_1']=\phi_R^\ast([c_1'])=\phi_R^\ast([c_2'])=[c_2']=[c_2].\]
To prove the remaining claim, we fix a $1$-cycle $c$ in $\nerve\mathcal{S}$. 
$c$ decomposes into a collection of ``simple'' loops, that is, loops in which every vertex is traversed only once.
There are three possibilities for such a loop: Its support lies (1) entirely in $M_R^{\mathrm{ext}}$,
(2) entirely in $\mathcal{S}\setminus M_R$, or (3) contains elements of both $M_R$ and 
$\mathcal{S}\setminus M_R^{\mathrm{ext}}$. In the latter case, both $A$ and $B$ must be in the support
(because $A$ and $B$ are the only entry points into $M_R^{\mathrm{ext}}$ from the outside by the guarding principle),
and the loop splits into two parts at $A$ and $B$, one that is supported entirely by $M_R^{\mathrm{ext}}$
and one entirely supported by $\mathcal{S}\setminus M_R$.
See also Figure~\ref{fig:cycle_transform}.

We construct $c'$ from $c$ as follows: Iterating over all loops of $c$, we remove loops of type (1),
and leave loops of type (2) unchanged. For loops of type (3), we replace the subpath
supported by $M_R^{\mathrm{ext}}$ with the edge $AB$. Note that after these replacements,
$c'$ is indeed supported by vertices in $\mathcal{S}\setminus M_R=\mathcal{S}_R$. 

It remains to prove that $[c]=[c']$.
We show that every loop transformation yields a homologous cycle.
For type (1), note that the loop is a cycle in $\nerve M_R^{\mathrm{ext}}$. 
By Lemma~\ref{lem:extended_members_homology}, such a cycle is 
null-homologous, so removing the loop does not change
the homology type. For type (2), there is nothing to do.
For type (3), consider the subpath of the loop inside $M_R^{\mathrm{ext}}$;
if it consists only of the edge $AB$, the loop remains unchanged.
Otherwise, the concatenation of the path with the edge $AB$
induces a cycle in  $\nerve(M_R^{\mathrm{ext}})$. 
Using Lemma~\ref{lem:extended_members_homology}, the loop is null-homologous,
thus we can one path by the other without changing the homotopy type.
\end{proof}

We complete the proof of the first main statement next.

\begin{figure}[htb]
\centering
\includegraphics[width=12cm]{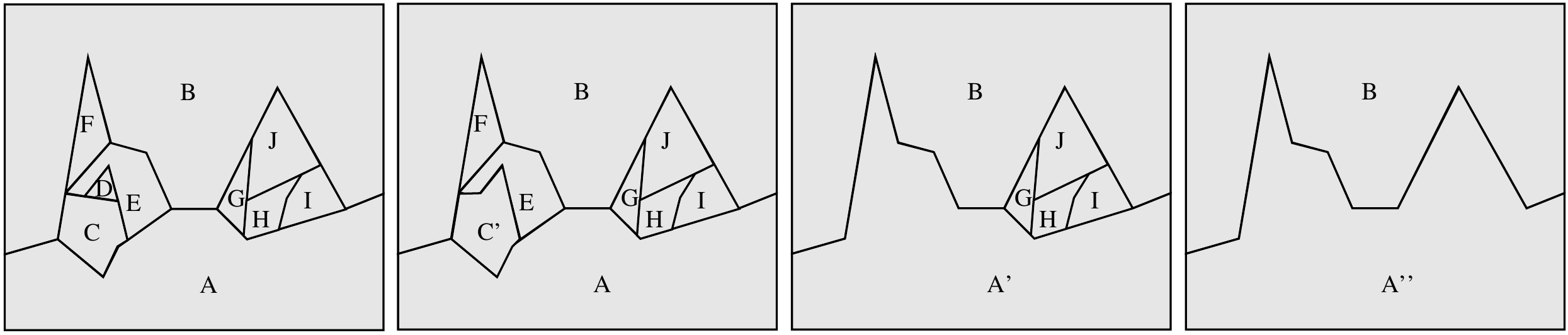}
\caption{Illustration of the transformation in the proof of Theorem~\ref{thm:tisk_theorem}.
The situation is a slight variation of Figure~\ref{fig:members}. We have three surrounded regions,
two induced by $(A,B)$, and one by $(C,E)$ and they are all filled. We first remove
the simple surrounded region induced by $(C,E)$. We charge the hole when removing the region
to one of the surrounders, say $C$, transforming it into $C'$ (2nd figure). This turns
the left surrounded region of $(A,B)$ simple as well. We perform the same operation twice
and transform $A$ into $A'$ (3rd figure) and finally into $A''$ (4th figure). After that,
the collection has no more surrounded region.}
\label{fig:eating_surrounded_regions}
\end{figure}

\begin{theorem}
\label{thm:tisk_theorem}
Let $\mathcal{S}$ be a collection of tisks in generic position such that every surrounded region is filled.
Then $H_1(\nerve(\mathcal{S}))=H_1(|\mathcal{S}|)$.
\end{theorem}
\begin{proof}
We assume that $\mathcal{S}$ induces at least one surrounded region.
In that case, it must also induce a simple one, and we let $R$ denote it
and $(A,B)$ be the pair in $\mathcal{S}$ that causes $R$.
By Lemma~\ref{lem:simple_iso_lemma}, because $R$ is filled, we have
$H_1(\nerve(\mathcal{S}))=H_1(\nerve(\mathcal{S}_R))$.
Now, we set $A':=A\cup R$ and consider $\mathcal{S}_R'$,
where we replace $A$ with $A'$ (Figure~\ref{fig:eating_surrounded_regions}). 
We note that
$\nerve \mathcal{S}_R'=\nerve \mathcal{S}_R$, because the extension
of $A$ did not introduce any new connection. This implies in particular that
$H_1(\nerve(\mathcal{S}))=H_1(\nerve(\mathcal{S}_R'))$
Moreover, $|\mathcal{S}_R'|=|\mathcal{S}|$, because
$A'$ occupies exactly the space that has been occupied by the members of $R$.
Finally, the surrounded regions of $\mathcal{S}_R'$
are equal to the surrounded regions of $\mathcal{S}$, except that the simple region
$R$ was removed. 
In other words, representing the surrounded region by a forest where $R_1$ is an ancestor of $R_2$
if $R_1\subseteq R_2$, the surrounding forest of $\mathcal{S}_R'$ equals the surrounding forest
of $\mathcal{S}$ with one leaf removed.

By iterating this construction, we find a collection $\mathcal{S}^\ast$ of tisks such that
$H_1(\nerve(\mathcal{S}))=H_1(\nerve(\mathcal{S}^\ast))$, $|\mathcal{S}|=|\mathcal{S}^\ast|$,
and $\mathcal{S}$ does not induce any surrounding region.
The last property, however, implies that all pairwise intersections are contractible.
Therefore, Lemma~\ref{lem:edges_suffice_for_nerve_theorem} applies and we can use the Nerve Theorem, 
which states that $\nerve(\mathcal{S}^\ast)$ is homotopically equivalent to $|\mathcal{S}^\ast|$.
In particular, their homology groups are isomorophic. Putting everything together, we have that
\[H_1(\nerve(\mathcal{S}))=H_1(\nerve(\mathcal{S}^\ast))=H_1(|\mathcal{S}^\ast|)=H_1(|\mathcal{S}|)\]
\end{proof}

Finally, we consider filtrations of tisks according to the following definition.

\begin{definition}\label{def:tisk-filtration}
A family of tisks $(S^\alpha)_{\alpha\geq 0}$ is called a \emph{tisk-inclusion}
if $S^\alpha\subseteq S^{\alpha'}$ for $\alpha\leq \alpha'$.
A collection of tisk-inclusions $(\mathcal{S}^\alpha)_{\alpha\geq 0}=(\{S_1^\alpha,\ldots,S_N^\alpha\})_{\alpha\geq 0}$
is called a \emph{tisk-filtration}. For each $\alpha\geq 0$, it defines a collection of tisks
at scale $\alpha$. 
A tisk-filtration is called \emph{sane} if for every $\alpha$, 
the collection of tisks is interior-disjoint, in generic position, and all its surrounded regions are filled.
\end{definition}

For a sane tisk-filtration $(\mathcal{S}^\alpha)_{\alpha\geq 0}$,
Theorem~\ref{thm:tisk_theorem} establishes an isomorphism between
the first homology group of $\nerve(\mathcal{S}^\alpha)$ and $|\mathcal{S}^\alpha|$ for every $\alpha>0$.
On the other hand, for $\alpha_1\leq\alpha_2$, we have natural inclusions from 
$|\mathcal{S}^{\alpha_1}|$ to $|\mathcal{S}^{\alpha_2}|$ as well as inclusions on their nerves.
We show next that these inclusions commute with the isomorphisms. 

\begin{theorem}\label{thm:connection}
For a sane tisk-filtration $(\mathcal{S}^\alpha)_{\alpha\geq 0}$ and $\alpha_1\leq\alpha_2$,
the diagram
\[
\xymatrix{
H_1(|\mathcal{S}^{\alpha_1}|) \ar @{^{(}->}[r] & H_1(|\mathcal{S}^{\alpha_2}|) \\
H_1(\nerve(\mathcal{S}^{\alpha_1})) \ar[u]^{\psi^\ast} \ar @{^{(}->}[r] &  H_1(\nerve(\mathcal{S}^{\alpha_2})) \ar[u]^{\psi^\ast}
}
\]
commutes, where $\psi^\ast$ is the isomorphism as constructed in Theorem~\ref{thm:tisk_theorem}.
\end{theorem}

\begin{proof}
The proof is lengthy and requires
us to study the isomorphisms induced by the Nerve theorem, in a similar spirit as in~\cite{co-towards}.
We let $R^{(1)}_1,\ldots,R^{(1)}_{h_1}$ denote the the surrounded regions of $\mathcal{S}^{\alpha_1}$
and $R^{(2)}_1,\ldots,R^{(2)}_{h_2}$ those of $\mathcal{S}^{\alpha_2}$. Clearly, $h_1\leq h_2$
and we can label the surrounded regions such that $R^{(1)}_i=R^{(2)}_i$ for $1\leq i\leq h_1$.
Recall that in the construction of $\psi^\ast$ (Theorem~\ref{thm:tisk_theorem}), 
we repeatedly remove tisks from surrounded regions and charge their area to one of its surrounding tisks.
After removing all surrounded regions, we arrive at a collection $(\mathcal{S}^{\alpha_i})^\ast$
from $\mathcal{S}^{\alpha_i}$ with the same underlying space, and a simplicial map 
\[\phi_i: \nerve(\mathcal{S}^{\alpha_i})\rightarrow \nerve((\mathcal{S}^{\alpha_i})^\ast),\]
which induces an isomporphism $\phi_i^\ast$ between the $1$-st homology groups.
Moreover, because $(\mathcal{S}^{\alpha_i})^\ast$ has no surrounded regions, the Nerve theorem defines an isomorphism 
\[\theta_i^\ast: H_i(\nerve (\mathcal{S}^{\alpha_i})^\ast)\rightarrow H_i(|\mathcal{S}^{\alpha_i})^\ast|).\]
Finally, there is a natural simplicial map
\[\gamma: \nerve((\mathcal{S}^{\alpha_1})^\ast)\rightarrow \nerve((\mathcal{S}^{\alpha_2})^\ast)\]
where we map a tisk that is not surrounded at $\alpha_1$, but surrounded at $\alpha_2$
to one of its surrounding tisks (in other words, $\gamma$ can be seen as the restriction
of $\phi_2$ to $\nerve((\mathcal{S}^{\alpha_1})^\ast)$).
Putting everything together, we have the following diagram
\begin{eqnarray}
\label{eqn:big_diagram}
\xymatrix{
H_1(|\mathcal{S}^{\alpha_1}|) \ar[r]^{\subseteq} & H_1(|\mathcal{S}^{\alpha_2}|) \\
H_1(|(\mathcal{S}^{\alpha_1}|)^\ast) \ar[r]^{\subseteq} \ar @{=}[u] & H_1(|(\mathcal{S}^{\alpha_2})^\ast|)  \ar @{=}[u] \\
H_1(\nerve (\mathcal{S}^{\alpha_1})^\ast) \ar[r]^{\gamma^\ast} \ar[u]^{\theta_1^\ast} & H_1(\nerve (\mathcal{S}^{\alpha_2})^\ast)  \ar[u]^{\theta_2^\ast}\\
H_1(\nerve \mathcal{S}^{\alpha_1}) \ar[r]^{\subseteq} \ar [u]^{\phi_1^\ast} & H_1(\nerve \mathcal{S}^{\alpha_2})  \ar[u]^{\phi_2^\ast}\\
}
\end{eqnarray}
where $\subseteq$ means that the corresponding map is induced by inclusion.
We show that every square in this diagram commutes. This is immediately clear for the upper square
from the top. For the lower square, we observe that the corresponding diagram on the simplicial level
already commutes, by definition of $\gamma$.

For the middle square, we need to investigate the isomorphism $\theta_i^\ast$ in more detail. 
(Note that we cannot apply Theorem~\ref{thm:nerve_inclusion_commute} because $\gamma^\ast$ is not induced by inclusion)
We still want to follow the approach from
\cite[Lemma 3.4]{co-towards}, \cite[Sec. 4G]{hatcher}.
A technical difficulty is that these results require \emph{open} covers of the underlying space, while we cover with
closed spaces. However, we can just replace any tisk in $(\mathcal{S}^{\alpha_i})^\ast$ by an offset of itself with a sufficiently small $\eps$-value
and restrict the offset to the underlying space. This yields an open cover 
$\mathcal{U}^{\alpha_i}$, with same underlying space and the same nerve
as $(\mathcal{S}^{\alpha_i})^\ast$~-- this is possible because we assume our tisks to be bounded by finitely many algebraic sets which rules out pathological
cases where tisks come arbitrary close to each other without intersecting. With that, we have the diagram
\begin{eqnarray*}
\xymatrix{
H_1(|(\mathcal{S}^{\alpha_1}|)^\ast) \ar[r]^{\subseteq} & H_1(|(\mathcal{S}^{\alpha_2})^\ast|)\\
H_1(|\mathcal{U}^{\alpha_1}|) \ar[r]^{\subseteq} \ar @{=}[u] & H_1(|\mathcal{U}^{\alpha_2}|)  \ar @{=}[u] \\
H_1(\nerve(\mathcal{U}^{\alpha_1})) \ar[u]^{\theta_1^\ast} \ar[r]^{\gamma^\ast} &  H_1(\nerve(\mathcal{U}^{\alpha_2})) \ar[u]^{\theta_2^\ast}\\
H_1(\nerve (\mathcal{S}^{\alpha_1})^\ast) \ar[r]^{\gamma^\ast} \ar @{=}[u] & H_1(\nerve (\mathcal{S}^{\alpha_2})^\ast)  \ar @{=}[u]\\
}
\end{eqnarray*}

where the upper and lower square obviously commute, and we only need to show that the middle square commutes.
Set $X_i:=|\mathcal{U}^{\alpha_i}|$ for convenience,
and let $U^{(i)}_1,\ldots,U^{(i)}_{n_i}$ denote the elements in the open cover $\mathcal{U}^{\alpha_i}$.
Let $\Delta^{n_i-1}$ denote the standard simplex of dimension $n_i-1$.
We define a space $\Delta X_i\subseteq X_i\times\Delta^{n_i-1}$ as follows:
Any non-empty subset $\sigma\subset\{1,\ldots,n_i\}$ defines a simplex $[\sigma]$ of $\Delta^{n_i-1}$ choosing the corresponding vertex set.
Also, $\sigma$ induces a (possibly empty) set $U^{(i)}_\sigma=\bigcap_{j\in\sigma}U^{(i)}_j$. We set
\[\Delta X_i:=\bigcup_{\emptyset\neq\sigma\subseteq\{1,\ldots,n_i\}} U^{(i)}_\sigma\times[\sigma].\]
Our next goal is to define a map that connects $\Delta X_1$ and $\Delta X_2$. Note first that
$\gamma:\nerve((\mathcal{S}^{\alpha_1})^\ast)\rightarrow \nerve((\mathcal{S}^{\alpha_2})^\ast)$
is defined through a vertex map from one nerve to the other, by identifying surrounded tisks
with one of their surrounders. By assigning indices and tisks, $\gamma$ can be encoded as a 
map $\gamma: \{1,\ldots,n_1\}\rightarrow \{1,\ldots,n_2\}$. Note that the individual tisks
are only growing when surrounded regions are eliminated; therefore, we have that
$U^{(1)}_k \subseteq U^{(2)}_{\gamma(k)}$. $\gamma$ also extends to a surjective map
from $\Delta^{n_1-1}$ to $\Delta^{n_2-1}$ in a natural way, and we have that
$U^{(1)}_\sigma\subseteq U^{(2)}_{\gamma(\sigma)}$. Therefore, the map
\[\xi: \Delta X_1\rightarrow \Delta X_2\]
which maps $(x,v)$ to $(x,\gamma(v))$ is well-defined.

Finally, let $\Gamma_i$ denote the barycentric subdivision of $\nerve \mathcal{U}^{\alpha_i}$
and note that $\gamma$ also extends to a map $\gamma':\Gamma_1\rightarrow\Gamma_2$ in a natural way.
Now we consider the following diagram

\[
\xymatrix{
X_1 \ar @{^{(}->}[r] & X_2\\
\Delta X_1 \ar[r]^\xi \ar[u]^{p_1}\ar[d]^{q_1} & \Delta X_2 \ar[u]^{p_2} \ar[d]^{q_2}\\
\Gamma_1 \ar[r]^{\gamma'} & \Gamma_2
}
\]
where $p_i$ is the natural projection from $\Delta_i$ to $X_i$
and $q_i$ is the map obtained by contracting every $U^{(i)}_\sigma$ to a single point, say $x_0$.
Both squares commute: For the first square, this is immediately clear,
because a point $(x,v)\in \Delta X_1$ is mapped to $x\in X_2$, regardless of how to follow the diagram.
For the second square, let $(x,v)\in\Delta X_1$ and note that $q_2(\xi(x,v))=q_2(x,\gamma'(v))=(x_0,\gamma'(v))$
and $\gamma'(q_1(x,v))=\gamma'(x_0,v)=(x_0,\gamma'(v))$.

We consider the diagram
\[
\xymatrix{
C_k(\Gamma_1) \ar[r]^{\gamma'_p} & C_k(\Gamma_2)\\
C_k(\nerve \mathcal{U}^{\alpha_1}) \ar[u]^{h_1} \ar[r]^{\gamma_p} & C_k(\nerve \mathcal{U}^{\alpha_2} \ar[u]^{h_2})
}
\]
of $k$-chain groups, where $\gamma_p$, $\gamma'_p$ are chain maps for dimension $p$ induced by $\gamma$ and $\gamma'$ (see~\cite[p.72]{munkres}). 
$h_i$ is the chain map defined by mapping a $k$-simplex $\sigma\in \nerve \mathcal{U}^{\alpha_1}$
to the chain of $k$-simplices that are incident to the vertex $\hat{\sigma}$ representing $\sigma$ in $\Gamma_i$.
It is again straight-forward to see that the diagram commutes: Fix a $k$-simplex $\sigma$. If $\sigma$ is contracted,
that is, two of its boundary vertices are identified, every simplex of $\Gamma_1$ incident to $\hat{\sigma}$ 
is also contracted. It follows that  $\gamma'_p(h_1(\sigma))=0=h_2(\gamma_p(\sigma))$. On the other hand,
of $\sigma$ is not contracted, every simplex in the barycentric subdivision that is incident to $\hat{\sigma}$
is mapped to a non-trivial $k$-chain, and it is easy to verify that $\gamma'_p(h_1(\sigma))=h_2(\gamma_p(\sigma))$
also in this case.

Summarizing the previous two steps, we have the commutative diagram

\[
\xymatrix{
H_1(X_1) \ar[r]^{i} & H_1(X_2)\\
H_1(\Delta X_1) \ar[r]^\xi \ar[u]^{p_1^\ast}\ar[d]^{q_1^\ast} & H_1(\Delta X_2) \ar[u]^{p_2^\ast} \ar[d]^{q_2^\ast}\\
H_1(\Gamma_1) \ar[r]^{\gamma'} & H_1(\Gamma_2)\\
H_1(\nerve \mathcal{U}^{\alpha_1}) \ar[u]^{h_1^*} \ar[r]^\gamma & H_1(\nerve \mathcal{U}^{\alpha_2} \ar[u]^{h_2^\ast})
}
\]

According to~\cite[Prop.4G.2 and 4G.3]{hatcher}, the maps $p_i$ and $q_i$ are isomorphisms, 
and according to~\cite[Thm.13.3 and 17.2]{munkres}, $h_i$ is an isomorphism. It follows that
with $\theta_i^\ast:=p_i^\ast\circ (q_i^\ast)^{-1}\circ h_i^\ast$, the middle square of (\ref{eqn:big_diagram})
commutes, which completes the proof.
\end{proof}

We apply the previous theorem on the case of restricted offset filtrations and prove Theorem~\ref{thm:1-barcode_equivalence}:

\begin{theorem}
For convex polygonal sites in $\R^2$, the $0$- and $1$-barcode of the restricted nerve filtration are equal to the $0$- and $1$-barcode
of the offset filtration, respectively.
\end{theorem}
\begin{proof}
For the $1$-barcode, it is enough to show that the restricted offset filtration yields a sane tisk-filtration 
according to Definition~\ref{def:tisk-filtration}.
By assumption, the restricted offsets are in generic position and interior-disjoint. It remains to show that for any $\alpha$,
any surrounded region is filled. 

Consider a region $R$ surrounded by $A$ and $B$ with boundary curves $a$ and $b$. Let $v_1$, $v_2$ denote the points
on the boundary of $R$ that lie in $a\cap b$. 
Assume wlog that $d(v_1)\leq d(v_2)=:w$. We have to show that the union of the restricted $w$-offset sites cover $R$.
For that, it suffices to show that the unrestricted $w$-offsets of $A$ and $B$ cover $R$, what we show next.
It can easily be seen that the bisector of $A$ and $B$ has a segment within $R$
that connects $v_1$ and $v_2$. With Theorem~\ref{thm:pseudodisk_maintext}, we have that $d(x)\leq w$ for all $x$ on that bisector segment.
Moreover, for any $x$ on the part of $a\setminus b$ that bounds $R$, we must have $d(x)\leq w$ as well. Combining these two properties,
the ``half-region'' of $R$ bounded by $a\setminus b$ and the bisector segment satisfies $d(x)\leq w$ on its boundary and by convexity
of the distance function, $d(x)\leq w$ in the whole region. Applying the same argument on the other half-region, we get the result.

The result for the $0$-barcode is obtained by proving that an analogue version of Lemma~\ref{lem:simple_iso_lemma}
also holds for $0$-homology. The proof for that is similar, but simpler than for the $1$-homology case. 
We omit further details.
\end{proof}

\section{Pseudo-disk property in three dimensions}
\label{app:pseudodisk}

Recall the following famous result on the intersection of Minkowski sums of convex objects~\cite{klps-union}, \cite[Thm.13.8]{dutch}
\begin{theorem}
Let $P_1$, $P_2$ be two convex interior-disjoint polygons and let $R\subset\R^2$ be another convex object. Then, $\partial(P_1\oplus R) \cap \partial(P_2\oplus R)$
consists of at most $2$ points.
\end{theorem}

(Partial) extensions of the theorem have been given, for instance, in~\cite{as-pipes}.
We will restrict our attention to the case that $R$ is the unit ball.
In this appendix, we will prove to the following 
generalization to the three-dimensional case~--we are indebted to Micha Sharir who sketched 
the basic construction idea of the proof~\cite{ShPC14}.
\begin{theorem}
Let $P_1$, $P_2$ be two convex disjoint polyhedra in $\R^3$ 
in general position and let $B$ be the unit ball. 
Then, $\partial(P_1\oplus B) \cap \partial(P_2\oplus B)$
is either empty, a single point, or homeomorphic to a closed cycle.
\end{theorem}

\begin{proof}
Since $P_1$ and $P_2$ are convex and interior-disjoint, there is a plane $h$ that separates the interior of the two sets.
We can assume w.l.o.g. that $h$ passes through the origin and that $h$ equals to $xy$-plane after a suitable translation and rotation.
We assume wlog that $P_1$ is above and $P_2$ is below $h$.
We fix the unit sphere $U$. Any point on $U$ defines a direction $v$ in $\R^3$. 
The direction $v$ defines an ordering of all planes that are normal to $v$. We say the a plane $e$ is \emph{further than} a plane $e'$
if $e'$ comes after $e$ when we go in direction $v$. Let $e_1$ be the furthest plane normal to $v$ that intersects $P_1$. The same way, $e_2$ is the furthest plane for $P_2$.
We call $P_1$ \emph{more extreme} than $P_2$ in direction $v$ if $e_1$ is further than $e_2$. We call them \emph{equally extreme} if $e_1$ and $e_2$ coincide.
Clearly, $U$ partitions into $E_1$, the set of directions that are more extreme for $P_1$, $E_2$, the set of directions more extreme for $P_2$,
and $E_{12}$, the directions which are equally extreme.

First of all, note that $E_1$ and $E_2$ are non-empty,
because in the upwards direction $z$, $P_1$ is more extreme, and in the downwards direction $-z$, $P_2$ is more extreme.
We show that $E_{12}$ is a closed cycle. For that, note that $h$
cuts out the equator from $U$. Letting $v$ denote a vector on the equator, let $C$ be the plane spanned by $v$ and $z$.
The projections $\hat{P}_1$, $\hat{P}_2$ of $P_1$, $P_2$ to $C$ are convex sets in the plane. Fixing a direction $c$ in $C$, we can observe
that $c$ is more extreme for $P_1$ than $P_2$ if and only if $c$ is more extreme for $\hat{P}_1$ than for $\hat{P}_2$,
and the same is true for equally extreme points.
For the planar case, however, it is known that there are precisely $2$ equally extreme directions, cutting the unit circle into two parts,
one of which gives all directions where $P_1$ is more extreme, and one where $P_2$ is more extreme. When varying $v$, the equally extreme
directions vary continuously, so the family of all such planes traces out a cycle of directions where $P_1$ and $P_2$
are equally extreme, cutting $U$ into two disks, one more extreme for $P_1$ and one more extreme for $P_2$.

Now, consider the set $K:=\partial(P_1\oplus B \cap P_2\oplus B)$. $K$ is either empty, a single point, or a $3$-dimensional convex shape.
We assume that it is $3$-dimensional, since the statement follows easily otherwise.
We consider the $\eps$-expansion $K_\eps:=K\oplus B_\eps$ of $K$, where $\eps>0$ arbitrary. 
Every point on $\partial K_\eps$ originates from a unique point on $\partial K$; 
we let $\phi:\partial K_\eps\rightarrow \partial K$ denote the corresponding map.
Also, $\partial K_\eps$ is differentiable. Therefore, every point on the boundary has a unique tangent plane and an associated
outward normal vector. Vice versa, every direction appears as an outward normal vector exactly once because $K_\eps$ is convex.
This implies the existence of a homeomorphism $\psi:U\rightarrow \partial K_\eps$. In particular, the cycle $E_{12}$
maps to a cycle $\gamma=\psi(E_{12})$ on $\partial K_\eps$.
Furthermore, for $p\in\partial K$, we call $v\in U$ a \emph{supporting direction} for $p$ if $K$ is completely contained in the half-space 
that contains $p-v$ and is bounded by the plane through $p$ that is normal to $v$. Then, $v$ is a supporting direction for $p$ if and only if there exists 
a $p'\in \partial K_\eps$ such that $\phi(p')=p$ and $v$ is the outward normal vector of $p'$.

Set $I:=\partial(P_1\oplus B) \cap \partial(P_2\oplus B)$.
Remember that we want to show that $I$ is empty, a point, or a cycle. 
We argue that $I=\phi(\gamma)$. Let $p\in I$.
Therefore, by the basic properties of Minkowski sums, 
there exists a supported direction $v_1$ by $p$ such that $p=p_1+v_1$ with $p_1\in P_1$
the most extreme point in $v_1$ direction. Since $p\in (P_2\oplus B)$ as well, 
it follows that $P_1$ is not more extreme than $P_2$ in direction $v_1$.
The same way, $p=p_2+v_2$, with $p_2\in P_2$ and $v_2$ also a supported direction by $p$,
and $P_2$ is not more extreme than $P_1$ in direction $v_2$.
Now, there exist points $p_1',p_2'\in\partial K_\eps$ with $\phi(p_1')=\phi(p_2')=p$
and the outward normal vectors of $p_1'$, $p_2'$ are $v_1$ and $v_2$, respectively.
Since $P_1$ is not more extreme than $P_2$ in direction $v_1$, we have that 
$p_1'\in \psi(E_2)\cup\gamma$. The same way, $p_2'\in\psi(E_1)\cup\gamma$.
If either of the points lies on $\gamma$, we are done. Otherwise,
since the preimage of $p$ under $\phi$ is connected, there is a path from $p_1'$
to $p_2'$ in the preimage, and this path has to cross $\gamma$. The intersection
point then asserts that $p\in\phi(\gamma)$. The opposite directions follows by
a similar argument.

Since $\phi$ is continuous, 
the previous argument shows that $I$ is closed curve on $\partial K$, but we have to exclude
the case that $I$ is self-intersecting.
Assume for a contradiction the presence of a self-intersection at $p\in\partial K$
and consider a sufficiently small neighborhood around $p$ on $\partial K$.
By assumption, there are at least $4$ arcs of $I$ leaving $p$;
for simplicity, assume that there are precisely $4$ arcs
(the argument easily extends to the case of higher singularities). 
These arcs cut the neighborhood into $4$ sectors.
Each sector is a part of $\partial K$ which either belongs to
$\partial(P_1\oplus B)$ or to $\partial(P_2\oplus B)$
(it can not belong to both, because then, the whole sector would belong to $I$).
By generic position, the sectors are alternating between 
$\partial(P_1\oplus B)$ and $\partial(P_2\oplus B)$.
Now, since $\partial(P_1\oplus B)$ is differentiable and $p$ is part of it,
there is a unique tangent plane $T_1$ at $p$ for $\partial(P_1\oplus B)$.
The same applies for the set $\partial(P_2\oplus B)$, and we let $T_2$ denote
the tangent plane with respect to this set.
Because of the alternating sector property, it follows that $T_1=T_2$.
Let $v$ be the direction normal to $T_1$. It follows that there exist
points $p_1\in P_1$ and $p_2\in P_2$ such that $p_1+v=p=p_2+v$.
This implies that $p_1=p_2$, contradicting the assumption that $P_1$
and $P_2$ are disjoint.
\end{proof}

\end{appendix}
\end{document}